\documentclass[a4paper,UKenglish]{lipics-v2019}
\usepackage{amssymb}
\usepackage{MnSymbol}
\usepackage{varioref}
\usepackage[all]{xy}
\usepackage{yfonts}
\usepackage{url}
\usepackage{graphicx}
\usepackage{color}
\usepackage{xcolor}
\xymatrixrowsep{1.5pc} 
\xymatrixcolsep{1.5pc}

\newcommand\establish[1]{\ensuremath{{\vartriangleleft}{\blacktriangleright}}}

\newcommand\setgen[1]{\ensuremath{#1}}
\newcommand\setA{\setgen{A}}

\newcommand\setC{\setgen{C}}
\newcommand\setD{\setgen{D}}

\newcommand\setU{\setgen{U}}

\newcommand\graphgen[1]{\ensuremath{#1}}

\newcommand\graphG{\graphgen{G}}
\newcommand\graphH{\graphgen{H}}

\newcommand\graphU{\graphgen{U}}

\newcommand\nodesgen[1]{\ensuremath{\graphgen{#1}_V}}
\newcommand\edgesgen[1]{\ensuremath{\graphgen{#1}_E}}
\newcommand\sourcegen[1]{\ensuremath{sc^{\graphgen{#1}}}}
\newcommand\targetgen[1]{\ensuremath{tg^{\graphgen{#1}}}}

\newcommand\nodesG{\nodesgen{G}}
\newcommand\nodesH{\nodesgen{H}}

\newcommand\edgesG{\edgesgen{G}}
\newcommand\edgesH{\edgesgen{H}}

\newcommand\sourceG{\sourcegen{G}}
\newcommand\sourceH{\sourcegen{H}}

\newcommand\targetG{\targetgen{G}}
\newcommand\targetH{\targetgen{H}}

\newcommand\graphlonggen[1]{\ensuremath{(\nodesgen{#1},\edgesgen{#1},\sourcegen{#1},\targetgen{#1})}}

\newcommand\graphlongG{\graphlonggen{G}}

\newcommand\katgen[1]{\ensuremath{\mathsf{#1}}}

\newcommand\katC{\katgen{C}}
\newcommand\katD{\katgen{D}}

\newcommand{\katbase}{\katgen{Base}}

\newcommand\katcat{\katgen{Cat}}

\newcommand\katgraph{\katgen{Graph}}
\newcommand\katGRAPH{\katgen{GRAPH}}

\newcommand\katmult{\katgen{Mult}}

\newcommand\katpar{\katgen{Par}}

\newcommand\katset{\katgen{Set}}

\newcommand\objectsgen[1]{\ensuremath{\katgen{#1}_V}}

\newcommand\objectsC{\objectsgen{C}}

\newcommand{\graphofcatgen}[1]{\ensuremath{gr(\mathsf{#1})}}

\newcommand\objectgen[1]{\ensuremath{#1}}

\newcommand\objecta{\objectgen{a}}
\newcommand\objectb{\objectgen{b}}

\newcommand\functorgen[1]{\ensuremath{\mathtt{#1}}}

\newcommand\calA{\ensuremath{\mathcal{A}}}

\newcommand\calI{\ensuremath{\mathcal{I}}}

\newcommand{\arity}[2][2]{\ensuremath{\alpha^{\sig{#1}[#2]}}}

\renewcommand\katgen[1]{\ensuremath{\mathbf{#1}}}

\newcommand{\katpred}{\ensuremath{\pmb{\Pi}}}
\newcommand{\setpred}{\setgen{\Pi}}
\newcommand{\katvar}{\katgen{Var}}
\newcommand{\katcontext}{\katgen{Ctxt}}
\newcommand{\kattype}{\katgen{Var}}
\newcommand{\katcarrier}{\katgen{Carr}}
\newcommand{\katstructure}{\katgen{Struct}}
\newcommand{\katsketch}{\katgen{Sketch}}
\newcommand{\katsemantics}{\katgen{Sem}}
\newcommand{\katmodels}{\katgen{Mod}}

\renewcommand\objectsgen[1]{\ensuremath{\katgen{#1}_{Obj}}}
\renewcommand\objectsC{\objectsgen{\katC}}
\newcommand\objectstype{\objectsgen{Var}}

\newcommand{\metasig}{\ensuremath{\Xi}}

\newcommand\featuregen[1]{\ensuremath{\mathtt{#1}}}
\newcommand\featuresibling{\featuregen{sibling}}
\newcommand\featurecomp{\featuregen{comp}}
\newcommand\featureid{\featuregen{id}}
\newcommand\featuremonic{\featuregen{monic}}
\newcommand\featurecsqu{\featuregen{csqu}}
\newcommand\featureprdtwo{\featuregen{prd(2)}}
\newcommand\featurepb{\featuregen{pb}}
\newcommand\featurefinal{\featuregen{final}}
\newcommand\featureprodn{\featuregen{prod(n)}}
\newcommand\featureprodtwo{\featuregen{prod(2)}}
\newcommand\featuremale{\featuregen{male}}
\newcommand\featureparent{\featuregen{parent}}
\newcommand\featuretermn{\featuregen{term(n)}}
\newcommand\featuretermtwo{\featuregen{term(2)}}

\newcommand\structureU{\ensuremath{\mathcal{U}}}
\newcommand\structureV{\ensuremath{\mathcal{V}}}

\newcommand\carrierU{\ensuremath{U}}
\newcommand\carrierV{\ensuremath{V}}

\newcommand{\predP}{\ensuremath{P}}

\newcommand{\typeX}{\ensuremath{X}}
\newcommand{\typeY}{\ensuremath{Y}}
\newcommand{\typeZ}{\ensuremath{Z}}

\newcommand{\substitutiona}{\ensuremath{\varphi}}

\newcommand{\expression}{\ensuremath{Exp}}
\newcommand{\expressiona}{\ensuremath{Exp_1}}
\newcommand{\expressionb}{\ensuremath{Exp_2}}

\newcommand\arityof{\ensuremath{\triangleright}}

\newcommand{\typemorphisma}{\ensuremath{\varphi}}
\newcommand{\typemorphismb}{\ensuremath{\psi}}

\newcommand\condexistential[4]{\ensuremath{(#1\rightarrow_{#2}\exists( #3: #4))}}
\newcommand\conduniversal[4]{\ensuremath{(#1\rightarrow_{#2}\forall( #3 : #4))}}
\newcommand\propimplication[2]{\ensuremath{(#1\rightarrow #2)}}
\newcommand\existsvia[3]{\ensuremath{\exists_{#1}(#2:#3)}}
\newcommand\forallvia[3]{\ensuremath{\forall_{#1}(#2:#3)}}

\newcommand\solution[4]{\ensuremath{#1\; \models^{#2} #3\arityof #4}}

\newcommand{\initialobject}{\ensuremath{0}}

\newcommand\interpretationa{\ensuremath{\iota}}
\newcommand\interpretationb{\ensuremath{\varrho}}

\newcommand\contextG{\ensuremath{G}}
\newcommand\contextK{\ensuremath{K}}
\newcommand\contextL{\ensuremath{L}}
\newcommand\contextR{\ensuremath{R}}

\newcommand\contextinterpretation[2]{\ensuremath{(#1,#2)}}

\newcommand\constraint[3]{\ensuremath{(#1\arityof #2,#3)}}

\newcommand\bindinga{\ensuremath{\beta}}

\newcommand\functorCtr{\functorgen{Cstr}}
\newcommand\functorCstr{\functorgen{Cstr}}
\newcommand\functorInt{\functorgen{Intr}}
\newcommand\functorMod{\functorgen{Mod}}

\newcommand{\contextmorphisma}{\ensuremath{\varphi}}

\newcommand{\matcha}{\ensuremath{\mu}}
\newcommand{\matchb}{\ensuremath{\nu}}

\newcommand{\sketchG}{\ensuremath{\mathfrak{G}}}
\newcommand{\sketchK}{\ensuremath{\mathfrak{K}}}
\newcommand{\sketchL}{\ensuremath{\mathfrak{L}}}
\newcommand{\sketchR}{\ensuremath{\mathfrak{R}}}
\newcommand{\sketchS}{\ensuremath{\mathfrak{S}}}

\newcommand{\setofconstraintsG}{\ensuremath{C^{\mathfrak{G}}}}
\newcommand{\setofconstraintsK}{\ensuremath{C^{\mathfrak{K}}}}
\newcommand{\setofconstraintsL}{\ensuremath{C^{\mathfrak{L}}}}
\newcommand{\setofconstraintsR}{\ensuremath{C^{\mathfrak{R}}}}

\newcommand\respace{\hspace{-2.8pt}}
\newcommand\lsemm{[ \respace [ } 
\newcommand\rsemm{]\respace ]} 
\newcommand\semm[1]{\ensuremath{\lsemm\, #1 \, \rsemm}}

\renewcommand\arity{\ensuremath{\alpha}}
\renewcommand\a{\ensuremath{\alpha}}

\newcommand\flar[3]{\ensuremath{#1\!:#2\rightarrow #3}}

\newcommand\dlar[3]{\ensuremath{#1:\xymatrix{{#2}\ar@{-|}[r]&{#3}}}} 

\newcommand{\la}{\langle}
\newcommand{\ra}{\rangle}

\newcommand{\isdef}{:=}

\newtheorem{open}[theorem]{Open issue}

\title{Logics of First-Order Constraints  \linebreak - A Category Independent Approach\footnote{This paper is a slight revision of the version prepared for CALCO 2019. Unfortunately, the paper presents our unconventional approach to logic a bit too abstract and the topic was not fitting that well in the focus of CALCO 2019 thus it was not included in the proceedings. We got, however the chance to give a talk within the CALCO Early Ideas stream. We included in this revision also some further ideas, conjectures, open issues and discussions.}}
\titlerunning{Logics of First-Order Constraints}
\author{Uwe Wolter}{Department of Informatics, University of Bergen, Norway}{Uwe.Wolter@uib.no}{}{}
\authorrunning{U.~Wolter}

\Copyright{Uwe Wolter}


\keywords{first-order logic, abstract model theory, institution, generalized sketch, algebraic specification, constraint, description logic}

\nolinenumbers 

\hideLIPIcs  

\EventEditors{tba}
\EventNoEds{0}
\EventLongTitle{Conference}
\EventShortTitle{Conference}
\EventAcronym{XXX}
\EventYear{2019}
\EventDate{}
\EventLocation{Somewhere}
\EventLogo{}
\SeriesVolume{}
\ArticleNo{} 

\begin{document}

\maketitle

\begin{abstract}
Reflecting our experiences in areas, like Algebraic Specifications, Abstract Model Theory, Graph Transformations, and Model Driven Software Engineering  (MDSE), we present a general, category independent approach to Logics of First-Order Constraints (LFOC). Traditional First-Order Logic,  Description Logic and the sketch framework are discussed as examples. 

We use the concept of institution \cite{Dia08,GB92} as a guideline to describe LFOC's. The main result states that any choice of the six parameters, we are going to describe, gives us a corresponding ``institution of constraints'' at hand. The ``presentations'' for an institution of constraints can be characterized as ``first-order sketches''. As a corresponding variant of the ``sketch-entailments'' in \cite{Mak97}, we finally introduce ``sketch rules'' to equip LFOC's with the necessary expressive power.
\end{abstract}

\section{Introduction}

The paper addresses a technological layer in traditional specification formalisms, like First-Order Logic (FOL), Universal Algebra and Category Theory, which is, not seldom, overseen and/or undervalued. In Software Engineering, in contrast, this layer is omnipresent in the form of a multitude of different kinds of ``diagrams'' (software models\footnote{We use the term ``model'' in two conflicting meanings: A ``software model'' is an abstract representation of certain aspects and properties of a software system while a ``model'' in logic is a structure conforming to a formal specification.}).

Nevertheless, this layer becomes also manifest in some concepts, we meet in traditional formalisms: ``Elementary diagrams'' in FOL, ``generators and defining relations'' in Group Theory and, as one of the paradigmatic examples, ``sketches'' in Category Theory \cite{BW90}.

There are different lines of motivation and challenges encouraging us to establish this layer as a subject of its own and start to develop an abstract and general account of it.

\noindent
\textbf{Diagram Predicate Framework (DPF):} Generalized sketches have been developed independently by Makkai, motivated by his work on an abstract formulation of Completeness Theorems in logic \cite{Mak97}, and a group in Latvia around Diskin, triggered by their work on data modeling \cite{CD96,Dis95}. Our further development of the generalized sketch approach, now called  ``Diagrammatic Predicate Framework (DPF)'', has been applied to a wide range of problems in Model Driven Software Engineering  (MDSE) \cite{DW08,RRLW09_JLAP,RRLW10_JLAP_NWPT09,RLGRW14_JFAC}. Applying DPF, we identified three deficiencies: (1) No operations. (2) Only ``atomic constraints'' and thus restricted expressiveness. (3) Too many auxiliary items when software models are formalized as sketches. Deficiency (1) is addressed in \cite{WolterDK18} while deficiencies (2) and (3) triggered the idea of arbitrary ``first-order constraints''\footnote{In Category Theory, the term ``diagram'' is used instead of ``constraint''. In finite product sketches, e.g., finite product diagrams are ``atomic constraints'' since finite products are primitives of the language.}. 

\noindent
\textbf{Category Theory:} Categories are graphs equipped with  composition and identities. All concepts and results in Category Theory are expressed in a language about graphs with composition and identity as the only primitives . The challenge is to put this understanding on a precise formal ground by developing a purely ``diagrammatic'' (graph-based) presentation of Category Theory where limits and colimits, for example, are described by ``first-order constraints'' on graphs.  
The vision is to have, one day, an interactive tool that allows our students to define concepts and to prove results in Category Theory based on pure ``diagrammatic reasoning''. Or, to formulate it differently: Let us present Category Theory in such a way, that ``diagram chasing'' becomes a precise and well-founded proof technique.

Only recently, in November 2020, we became aware of "the language of diagrams" introduced in \cite{Freyd76} and used in \cite{FreydS90} to present and define categorical concepts and carry out proofs in a diagrammatic manner. The logic of first-order constraints, we have been envisioned for category theory, seems to include, at the end, a variant of this "language of diagrams" of Freyd. We learned also about other interesting activities around the idea to turn "diagram chasing" into a precise and well-founded proof technique (see \cite{ochs2020favorite}).

\noindent
\textbf{Algebraic Specifications:} ``Partial algebras freely generated by a set of variables and a set of equations'' are a generalization of the concept of ``groups generated by a set of generators and a set of defining relations'' and have been a central technical tool in the small school on Partial Algebraic Specifications in former East-Germany  \cite{KaphengstReichel1971,ReichelHK1980,Rei87,Wol90,CGW95}). For decades, we intended to revise this concept and to lift it to a more general and broader level.

\noindent
\textbf{Graph Constraints:} Graph constraints, as presented in \cite{EEPT06} for example, have been a latent inspiration and reference during the development of the framework presented here. A thorough analysis of graph constraints and, especially, of the idea to "represent first-order logic using graphs" in \cite{Rensink04} will be a subject of future research.

\noindent
\textbf{Description Logic:}  Only when the CALCO 2019 version of this paper was about to be finished, we got acquainted with Description Logics \cite{BaaderHS2007}. We saw that the distinction between ``ABox'' and ``TBox'' correlates very well with our ideas. So, as a kind of ``a posterior motivation'', we hope that our approach may open a way to develop appropriate ``diagrammatic'' versions of Description Logics.

\noindent
\textbf{Content of the paper:}
To cover a wide range of applications, we decided to develop and present logics of first-order constraints (LFOC's) in a top-down manner - from abstract to more concrete. 
In the paper, we present the first stage of expansion on a very abstract level. We show how LFOC's can be defined in arbitrary categories. In the next stage of expansion we have to investigate, once in detail, what additional categorical infrastructure we need to have things like "closed formulas", substitutions, instances of "formulas" and similar features of traditional logics. Moreover, we have to work out all the examples in very detail.

We revise and extend, radically, the concepts and results in \cite{DW08}. Thereby, we use the concept of institution \cite{GB92,Dia08} as a guideline to describe our logics in a well-organized and incremental way. The main result states that any choice of the six parameters, we are going to introduce, namely base category, variable declarations, footprint, carriers, structures, and contexts, gives us a corresponding ``institution of constraints'' at hand. 

Base categories are discussed in Section \ref{sec:base-category} and variable declarations in Section \ref{sec:variables-Features-footprints}, respectively. To avoid misperceptions, we use in Subsection \ref{sec:syntax-variables-features-footprints} the term "feature" instead of "predicate" and the term "footprint" instead of "signature". So, a feature is given by a symbol/name and a variable declaration called its arity. A footprint \metasig\ is then just a collection of features. Based on a choice of possible carriers, we define \metasig-structures in Subsection \ref{sec:semantics-variables-features-footprints} in full analogy to traditional FOL. 

Section \ref{sec:first-order-expressions} presents the syntax and semantics of "first-order feature expressions". By a ``feature expression'' we mean something like a ``formula with free variables'' in traditional FOL. We do not consider them, however, as formulas, but rather as ``derived features''\footnote{In \cite{FreydS90} Freyd and Scedrov use the term "elementary predicate" for what we would consider a certain kind of "first-order feature expressions".}. Each feature expression is build upon a corresponding variable declaration that we also call its arity (since it is a "(derived) feature"). 

Generalizing concepts like ``set of generators'' in Group Theory, ``underlying graph of a sketch'' in Category Theory, ``set of individual names'' in Description Logics and ``underlying graph of a model'' in Software Engineering, we coin in Section \ref{sec:institutions-of-contraints} the concept ``context''. Contexts are our ''signatures'' in the sense of institutions.

As ''sentences'', in the sense of institutions, we introduce ``constraints'' in generalizing the corresponding concepts ``defining relation'' in Group Theory, ``diagram in a sketch'' in Category Theory, ``concept/role assertion'' in Description Logic and ``constraint'' in Software Engineering. 
''Models'', in the sense of institutions, are interpretations of contexts in \metasig-structures. Section \ref{sec:institutions-of-contraints} closes with the main theorem stating that any choice of the mentioned six parameters gives us a corresponding \textbf{institution of constraints} at hand.

To get, however, adequate ``logics of first-order constraints'', we have to take two steps more: 
Any institution provides, in a canonical way, a corresponding category of presentations as well as an extension of the ``model functor'' of the institution to this category of presentations. In Section \ref{sec:sketches} we outline this procedure for ``institutions of constraints'' where the presentations turn out to be nothing but the ``first-order sketches'', we have been looking for. In the final step in Section \ref{sec:sketch-rules}, we introduce ``sketch rules'', as a generalization of the ``sketch-entailments'' in \cite{Mak97}, to equip LFOC's with the necessary expressive power. We discuss concepts, like ``elementary diagrams'', ``axioms'' and ``deduction'', in the light of the new kind of logics, and close the paper with a short section concerning future research.

This paper presents only a starting point for the development of a framework of Logics of First-Order Constraints. Examples are not worked out in detail and many insights and ideas are only mentioned or sketched. Also the relation between LFOC's and traditional FOL's, for example, is not fully clarified yet and is not discussed. Nevertheless, we are convinced that the first building block of a general framework of LFOC's is here.

\section{Base Category}
\label{sec:base-category}

We start with some notational conventions: By \objectsC\ we denote the collection of objects of a category \katC\ and by $gr(\katC)$  the underlying graph of \katC.  We write often $\objecta\in\katC$, instead of $\objecta\in\objectsC$, if \objecta\ is an object in a category \katC. \(\katC(\objecta,\objectb)\) denotes the collection of all morphisms from object \objecta\ to object \objectb\ in \katC. 
That a graph \graphG\ is a subgraph of a graph \graphH, is denoted by $\graphG\sqsubseteq\graphH$. Analogously, $\katC\sqsubseteq\katD$, means that the category \katC\ is a subcategory of \katD.

To define a certain logic of first-order constraints, we first have first to choose a base category  comprising as well the basic syntactic entities as the  semantic domains of our logic. The base category  fixes, somehow, the ``conceptual universe'' we want to work within.
\begin{definition}[First parameter: Base Category]
	\label{def:first-parameter}
	The first parameter of a logic of first-order constraints (LFOC) is a chosen \textbf{base category} \katbase.
\end{definition}
A formalization of modeling techniques in Software Engineering, like ER diagrams, class diagrams and relational data models, for example, relies often on the category \katGRAPH\ of (directed multi) graphs and graph homomorphisms as a base category (compare \cite{DW08,RRLW09_JLAP,RRLW10_JLAP_NWPT09,RLGRW14_JFAC}). 
A (directed multi) graph $\graphG=\graphlongG$ is given by a collection \nodesG\ of vertices, a collection \edgesG\ of edges and functions \(\sourceG:\edgesG\to\nodesG\), \(\targetG:\edgesG\to\nodesG\) assigning to each edge its source and target vertex, respectively.  A graph \graphG\ is ``small'' if \nodesG\ and \edgesG\ are sets. 

A homomorphism \(\varphi:\graphG\to\graphH\) between two graphs is given by a pair of maps \(\varphi_V:\nodesG\to\nodesH\), \(\varphi_E:\edgesG\to\edgesH\) such that \(\sourceG;\varphi_V=\varphi_E,\sourceH\) and \(\targetG;\varphi_V=\varphi_E,\targetH\).

\katGRAPH\ comprises as well finite and small graphs as the (''very big'') underlying graphs of the category \katset\ of all sets and total maps, the category \katmult\ of all sets and multimaps (set-valued functions) and the category \katgraph\ of small graphs, for example. \katGRAPH\ is also the category of choice for a uniform description of  specification formalisms, like (the different variants) of First-Order Logic (FOL) and Universal Algebra, for example.
\begin{example}[Universal Algebra: Base category]
	\label{ex:universal-algebra:base-category}
	To illustrate our approach on the meta-level of formalisms, we will outline how the traditional formalism \textbf{''Many-Sorted Universal Algebra''} may be described. \katGRAPH\ is our base category in this example. 
	
\end{example}

\begin{example}[FOL: Base category]
	\label{ex:FOL:base-category}
	This example is located one abstraction (modeling) level below Example \ref{ex:universal-algebra:base-category}. We are not going to describe the formalism \textbf{''FOL with predicates only''}. Instead, we examine this variant of FOL on the level of first-order signatures and first-order structures. 
	Therefore, we choose the category \katset\ as the base category.
	The prototypical description logic \textbf{''Attributive Concept Language with Complements'' (ALC)} can be seen as a fragment of ``FOL with predicates only'' (see \cite{BaaderHS2007}).
\end{example}

\begin{example}[Category Theory: Base category]
	\label{ex:category-theory:base-category}
	Located on the same abstraction (modeling) level as Example \ref{ex:FOL:base-category}, we will outline a diagrammatic (graph-based) version of the theory of small categories, thus the category \katgraph\ of small graphs is our base category of choice.
\end{example}
\begin{open}[Other base categories]
\label{open:other-base-categories}
To validate that our approach is indeed universal, we should also find an example that goes beyond graphs. We could look at E-graphs \cite{EEPT06}, for example. This may also give a guideline how to extend our approach to arbitrary presheaf topoi!?

%
\end{open}

\section{Variables, Features and Footprints}
\label{sec:variables-Features-footprints}

\subsection{Syntax: Variables, Features and Footprints}
\label{sec:syntax-variables-features-footprints}
Traditionally, the construction of syntactic entities in logic, like terms, expressions and formulas, starts by declaring what variables can be used. Often, we assume an enumerable set of variables and then any term, expression or formula is based upon a chosen finite subset of this enumerable set of variables. Moreover, variable substitutions can be described by maps between finite sets of variables.
Abstracting from this traditional approach, we announce what kind of  ``variable declarations'' we want to use in our logic. 
\begin{definition}[Second parameter: Variables]
\label{def:second-parameter}
As the second parameter of a LFOC, we choose a small subcategory \kattype\ of the base category \katbase. We refer to the objects in \kattype\ as \textbf{''variable declarations''} while the morphisms in \kattype\ will be called \textbf{''variable substitutions''}.
\end{definition}

\begin{remark}[Signature extensions vs. generators]
\label{rem:signature-extensions-vs-generators}
CALCO-reviewer 2 stated that "it may be worth noting that this is a completely different view on variables compared to the institutional 'tradition', where variables generally depend on the notion of signature" (compare \cite{Dia08}). 

The mentioned "institutional tradition" reflects a standard technique in Universal Algebra and Algebraic Specifications used to describe the construction of free algebras by means of so-called signature extensions, i.e., the elements of a given $\Sigma$-algebra are added as auxiliary constants to the signature $\Sigma$ and then they are "syntactically translated" along a signature morphism \(\varphi:\Sigma\to\Sigma'\). 
As indicated in the introduction, we consider this technique as inadequate and kind of dubious. We insist to keep apart signatures and variables (or, more precisely, "generators of algebras"). The idea of "generators of algebras" is one of our main motivations to introduce "contexts" as a basic building block of LFOC's (see Section \ref{sec:contexts-and-sentence-functor}). 
\end{remark}

\begin{example}[Universal Algebra: Variables]
	\label{ex:universal-algebra:variables}
	Our ``variable declarations'' are ``graphs of variables'', i.e., we declare two kinds of variables: vertex variables and edge variables that are connecting vertex variables. 
	As \kattype\ one can choose an enumerable subcategory  of \katGRAPH. 	
	For this paper, we choose \kattype\ to be the full subcategory of \katGRAPH\ given by all finite graphs $\graphG=(\nodesG,\edgesG,\sourceG:\edgesG\to\nodesG,\targetG:\edgesG\to\nodesG)$ with \nodesG\ a finite subset of the set $\{pv,pv_1,pv_2,\ldots,xv,xv_1,xv_2,\ldots\}$ and \edgesG\ a finite subset of the set $\{pe,pe_1,pe_2,\ldots,xe,xe_1,xe_2,\ldots\}$.
	The variables starting with \(p\) will be preferably used to describe the arities of ``feature symbols'' (compare Definition \ref{def:third-parameter} and the corresponding examples). ``\(e\)'' in a variable name stands for ``edge'' while ``\(v\)'' refers to ``vertex''. 
\end{example}

\begin{example}[FOL: Variables]
	\label{ex:FOL:variables}
	''Variable declarations'' are traditionally just ``finite sets of variables''. 
	%
	We take \kattype\ to be the full subcategory of \katset\ given by all finite subsets of the set $\{p,p_1,p_2,\ldots,x,x_1,x_2,\ldots\}$.
	Officially, there are no variables in ALC. To describe, however, ALC as a fragment of FOL we choose the same variable declarations as in FOL but restrict variable substitutions to injective maps.
\end{example}

\begin{example}[Category Theory: Variables]
	\label{ex:category-theory:variables}
	As in Example \ref{ex:universal-algebra:variables}, we choose \kattype\ to be given by all finite graphs $\graphG=\graphlongG$ with \nodesG\ a finite subset of the set $\{pv,pv_1,pv_2,\ldots,xv,xv_1,$ $xv_2,\ldots\}$ and \edgesG\ a finite subset of $\{pe,pe_1,pe_2,\ldots,xe,xe_1,xe_2,\ldots\}$, but now considered as a full subcategory of \katgraph.
\end{example}

In traditional FOL the arity of predicates is described by natural numbers. To declare a binary predicate, for example, we have to declare a predicate symbol \texttt{P} with arity 2. This means, whenever \texttt{P} appears in an expression there are two adjoint ``positions'' for variables or values, respectively. Correspondingly, the semantics of \texttt{P} in a first-order structure with carrier set \setA\ is a subset of $\setA\times\setA$.

There is a bijection between the Cartesian product $\setA\times\setA$ and the set of all maps from a two-element set of variables (identifiers of positions), like $\{p_1,p_2\}$ for example, into the set \setA. Relying on this observation, we define arities by means of ``variable declarations''.
\begin{definition}[Third parameter: Footprint]
\label{def:third-parameter} The third parameter of a LFOC is a \textbf{footprint} $\metasig=(\setpred,\arity)$ over \kattype\ given
by a set $\setpred$ of \textbf{feature symbols} and a
	map \flar{\arity}{\setpred}{\objectstype}. For any
	$\predP\in \setpred$, the $\kattype$-object $\arity(\predP)$ is called the \textbf{arity of $P$}. We will often write $\a \predP$ for $\arity(\predP)$.
\end{definition}
\begin{remark}[Terminology: Footprint vs.\ signature]
	\label{rem:terminology-footprint-vs-signature}
	In most of our applications, footprints occur as \textbf{meta-signatures}, in the sense, that each specification formalism (modeling technique) is characterized by a certain footprint. Each of the formalisms Universal Algebra, Category Theory, and First-Order Logic is characterized by a certain footprint. The sketch data model in \cite{JRW02} corresponds to a certain footprint and so on.
	For footprints of the modeling techniques ``class diagrams'' and  ``relational data model'' we refer to \cite{Rutle2010phd,RRLW10_JLAP_NWPT09}.
	
	Until today, we used in all our MDSE papers the terms ``signature'' instead of ``footprint'' and ``predicate symbol'' instead of ``feature symbol''.
	This turned out to be a source for serious misunderstandings and misleading perceptions thus we decided to coin in this paper the new terms ``footprint'' and ``feature symbol'' instead.
\end{remark}
\begin{remark}[Dependencies between features]
\label{rem:dependencies-between-features}
In \cite{DW08} we introduced and elaborated ``dependencies'' between feature symbols. We considered categories \katpred\ of feature symbols together with arity functors \flar{\arity}{\katpred}{\kattype}. 
We drop  ``dependencies'' in this paper, to make things not too complicated in the beginning.

Extending Makkai's approach \cite{Mak97}, we worked in \cite{DW08} with categories \katpred\ of feature symbols, instead of just sets of feature symbols, and with arity functors \flar{\arity}{\katpred}{\kattype}, instead of just arity maps. Arrows between feature symbols represent dependencies between features. This allows us to reflect, already on the very basic level of feature symbols and thus prior to arities and semantics of features, that certain features depend on (are based upon) other features. As examples, one may express that both concepts "pullback" and "pushout" are based upon the concept "commutative square" and that the categorical concept "inverse image" depends on the concept "monomorphism". 

Any semantics of feature symbols has then to respect those dependencies. Dependency arrows are a tool to represent knowledge about and requirements on features prior to and independent of any kind of logic. Dependency arrows make somehow the framework of generalized sketches conceptual and structural round. It is may be worth to mention that the concept of "order-sorted algebra" can be seen as a somehow related idea where we do not impose, on the level of signatures, arrows between predicate symbols but between sort symbols \cite{GoguenMeseguer1992}.

In this first paper about LFOC's we drop dependency arrows due to, at least, three reasons:
(1) In order not to scare too many potential readers we do not want to deviate too much from the traditional first-order logic setting. (2) If we introduce dependencies between feature symbols, we should consequently describe to what extend and how they generate dependencies between feature expressions (introduced in Section \ref{sec:first-order-expressions}). On one side, this looks technically not fully trivial and, on the other side, such an effort has no relevance for our applications. (3) The requirements expressed by dependency arrows between feature symbols can be mimicked by the logic means we are going to introduce later.
\end{remark}
\begin{example}[FOL: Footprint]
\label{ex:FOL:footprint}
On this abstraction level, footprints correspond just to  traditional FOL signatures. Especially, features are just traditional predicates. In our sample footprint we declare a unary feature symbol \featuremale\ with arity $\{p\}$ and a tertiary feature symbol \featureparent\ with arity \(\{p_1,p_2,p_3\}\). 
	
A signature in ALC declares a set \(N_C\) of ``concept names'' and a set \(N_R\) of ``role names''. In view of footprints, this means to  declare a set \(N_C\) of feature symbols all with arity  $\{p\}$ and a set  \(N_R\) of feature symbols all with arity  $\{p_1,p_2\}$. A signature in ALC declares also a set \(N_O\) of ``individual names (nominals, objects)''. In our approach, those ``individual names'' are located in the so-called ``contexts'' (compare Example \ref{ex:FOL:sketches}).
\end{example}

\begin{example}[Category Theory: Footprint]
\label{ex:category-theory:footprint}
Categories are graphs equipped with a composition operation and an identity operation. At the present stage of expansion, we do not include operations (to be defined along the ideas from \cite{WolterDK18}) in our footprints thus we have to formalize composition and identity by means of features. In such a way, the footprint for the formalism ``Category Theory'' declares two feature symbols \featurecomp\ and \featureid\ with arities described in the following table.\\
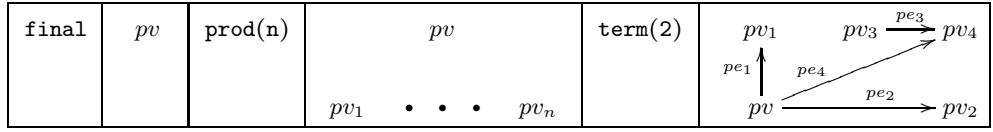
\begin{center}		
\begin{tabular}{|p{16mm}|p{35mm}|p{16mm}|p{20mm}|}
\hline
\featurecomp & $\xymatrix{pv_1 \ar[r]^{pe_1} \ar@/_1pc/[rr]_{pe_3} & pv_2 \ar[r]^{pe_2}  & pv_3}$ &
\featureid &
$\xymatrix{pv\ar@(r,d)[]^{pe} }$ 
\\
\hline
\end{tabular}
\end{center}

To force both existence and uniqueness of composition and  identities, respectively, we have to use ``sketch rules'', introduced in Definition \ref{def:sketch-rule}.  Also associativity of composition and neutrality of identities w.r.t.\ composition can be axiomatized by means of sketch rules.

We could extent the language of our formalism "Category Theory" to represent properties like monomorphism, commutative square, binary product and pullbacks, for example. That is, we could declare the following feature symbols with their corresponding arities:
\begin{center}
{\rm
\begin{tabular}{|p{16mm}|p{37mm}|}
\hline
\vspace{0.1mm} \textbf{Symbol} $P$& \vspace{0.1mm} \textbf{Arity} $\alpha P$ \\
\hline
\featuremonic & \hspace{3ex}
$\xymatrix{pv_1 \ar[r]^{pe_1}& pv_2}$ \\
\hline
\featurecsqu & \hspace{1ex} $\xymatrix{pv_1 \ar[r]^{pe_1} \ar[d]_{pe_3} & pv_2 \ar[d]^{pe_2}\\  pv_3 \ar[r]^{pe_4}& pv_4}$ \\
\hline
\featureprdtwo & $\xymatrix{& pv_1 \ar[dr]^{pe_2} \ar[dl]_{pe_1}\\   pv_2&& pv_3}$ \\
\hline
\featurepb & \hspace{1ex} $\xymatrix{pv_1 \ar[r]^{pe_1} \ar[d]_{pe_3} & pv_2 \ar[d]^{pe_2}\\  pv_3 \ar[r]^{pe_4}& pv_4}$ \\
\hline
\end{tabular}}
\end{center}

One of our main motivations to develop LFOC's was to be able to express the universal properties, defining those concepts only based on composition and identities, by means of first-order feature expressions introduced in Section \ref{sec:first-order-expressions}. 
\end{example}
\begin{example}[Universal Algebra: Footprint]
\label{ex:universal-algebra:footprint}
We don't want to recapitulate the inside that the formalism "Many-sorted Universal Algebra (with equations)" can be encoded by means of the formalism ``finite product sketches'' \cite{BW90}. We rather intend to reflect directly the traditional step-by-step approach in Universal Algebra: introduce sorts -- construct finite Cartesian products of sorts -- add operations -- introduce variables -- construct terms -- introduce equations. 

So, we declare a feature \featurefinal\ to denote empty products and a feature \featureprodn\ to denote n-ary products for each \(n\geq 1\).
Note, that we hide, in contrast to finite product sketches, the "auxiliary" projection arrows. This reflects the traditional presentation of signatures in Universal Algebra, where the projections are not present in the syntax but only come to the surface in the semantics (see the short remark on "(semantical) induced sketch rules" in Section \ref{sec:sketch-rules}).

Moreover, we declare a feature \featuretermn\ describing the construction of terms for n-ary operation symbols for each \(n\geq 1\) . The table shows the case \(n=2\), where \(pe_1,pe_2\) are the ``subterms'', \(pe_3\) is the ``operation symbol'' and \(pe_4\) is the ``generated term'' \(pe_3\la pe_1,pe_2\ra\). It would be natural, to express the requirement that \(pv_3\) is a binary product of \(pv_1\) and \(pv_2\) by a dependency between \featureprodtwo\ and \featuretermtwo\ (see Remark \ref{rem:dependencies-between-features}). Since, we don't include, at the moment, dependencies, we may declare later a corresponding "axiom rule" (see Section \ref{sec:sketch-rules}). As usual in Universal Algebra, we hide the "auxiliary arrow" \(\la pe_1,pe_2\ra:pv\to pv_3\) and the information that \(pe_4\) is the composition of \(\la pe_1,pe_2\ra\) and \(pe_3\) (see Example \ref{ex:universal-algebra:semantics}).
\begin{center}
\begin{tabular}{|p{8mm}|p{7mm}|p{11mm}|p{32mm}|p{11mm}|p{34mm}|}
	\hline
	\featurefinal & \hspace{0ex}
	$\xymatrix{pv}$ &
	\featureprodn & $\xymatrix{& pv \\
	pv_1\ar@{}[rr]|{\bullet\quad\bullet\quad\bullet}&& pv_n}$ &
	\featuretermtwo & $\xymatrix{pv_1 & pv_3\ar[r]^{pe_3} & pv_4\\
                	pv\ar[rr]^(.6){pe_2}\ar[u]^{pe_1}\ar[urr]^(.3){pe_4}&& pv_2}$ \\
	\hline
\end{tabular}
\end{center}
To formalize the step-by-step approach precisely, we should actually use as arities graphs typed over a type graph \graphgen{UA} consisting of two vertices \(\mathit{Sort},\mathit{CProd}\) and two edges \(op,term\) from \(\mathit{CProd}\) to \(\mathit{Sort}\). The point is, that the vertex \(pv\) in \featurefinal, \featureprodn, \featuretermtwo\ as well as the vertex \(pv_3\) in \featuretermtwo\ are of type \(\mathit{CProd}\) while the vertices \(pv_1,\ldots,pv_n\) in \featureprodn\ and the vertices \(pv_1,pv_2\) in \featuretermtwo\ are of type \(\mathit{Sort}\). Moreover, it would be necessary to consider a sequence of ``language extensions'': (1) Only ``type node'' \(\mathit{Sort}\). (2) Add ``type node''  \(\mathit{CProd}\) and features \featurefinal, \featureprodn. (3) Add ``type edge'' \(op\). (4) Add ``type edge'' \(term\) and feature \featuretermn. Both topics - ``typed variable declarations'' and ``language extensions'' - will be subjects of future developments.
\end{example}

\subsection{Semantics: Variables, Features and Footprints}
\label{sec:semantics-variables-features-footprints}

To make things not too complicated,
we work here with the traditional semantics-as-interpretation paradigm, in contrast to \cite{DW08}, where we spelled out the semantics-as-instance paradigm (fibred semantics).
To define the semantics of variables and features, we have first to decide for (potential) carriers of structures.
\begin{definition}[Fourth parameter: Carriers]
\label{def:fourth-parameter}
As the fourth parameter of a LFOC, we choose a subcategory \katcarrier\ of \katbase\ of (potential) \textbf{carriers} of \;\metasig-structures.
\end{definition}
\begin{example}[FOL: Carriers]
\label{ex:FOL:carriers}
Here we choose just \(\katcarrier=\katbase=\katset\).  ALC considers also arbitrary sets as potential carriers only that they are called ``domains (of an interpretation)''.
\end{example}

\begin{example}[Category Theory: Carriers]
\label{ex:category-theory:carriers}
We could choose only those graphs that appear as underlying graphs of small categories. We will, however, not restrict ourselves and chose, in analogy to Example \ref{ex:FOL:carriers}, \(\katcarrier=\katbase=\katgraph\).
\end{example}
\begin{example}[Universal Algebra: Carriers]
\label{ex:universal-algebra:carriers}
The traditional approach choses the underlying graph \graphofcatgen{\katset} of the category \katset\ as the only (potential) carrier.

%
\end{example}

The semantics of a ``variable declaration'' \(\typeX\in\katvar\) relative to a chosen carrier \(\carrierU\in\katcarrier\) is just the set of all \textbf{''variable assignments''} (remind that \(\katvar\sqsubseteq\katbase\) and \(\katcarrier\sqsubseteq\katbase\))
\begin{equation}\label{eq:semantics-variable-declaration}
	\semm{\typeX}^\carrierU\isdef\katbase(\typeX,\carrierU).
\end{equation}
Structures for footprints are defined in full analogy to the definition of structures for signatures in traditional first-order logic.
\begin{definition}[Structures]
\label{def:structures}
A \metasig-structure $\structureU =(\carrierU,\setpred^\structureU)$ is given by an object \carrierU\ in \katcarrier, the \textbf{carrier of} \structureU, and a family $\setpred^\structureU=\{\semm{\predP}^\structureU\mid \predP\in\setpred\}$ of sets $\semm{\predP}^\structureU\subseteq\katbase(\a\predP,\carrierU)$ of \textbf{''valid interpretations''} of feature symbols \predP\ in \carrierU. 
\end{definition}
Homomorphisms are also defined in the usual and obvious way that ``truth is preserved''.
\begin{definition}[Homomorphisms]
	\label{def-homomorphism}
	A \textbf{homomorphism} $\varsigma:\structureU\to\structureV$ between \metasig-structures is given by a morphism \flar{\varsigma}{\carrierU}{\carrierV} in \katcarrier\ such that $\interpretationa\in\semm{\predP}^\structureU$ implies $\interpretationa;\varsigma\in\semm{\predP}^\structureV$ for all feature symbols \predP\ in \setpred\ and all interpretations $\interpretationa:\a\predP\to\carrierU$.
	$$\xymatrix{
		&  \a\predP\ar[dl]_{\interpretationa\in\semm{\predP}^\structureU}
		\ar[dr]^{\interpretationa;\varsigma\in\semm{\predP}^\structureV}\ar@{}[d]|(.5)\Rightarrow
		\\
		\carrierU\ar[rr]^\varsigma
		&& \carrierV
	}$$
\end{definition}
Identities of carriers define identity homomorphisms and composition of homomorphisms is inherited from composition in \katcarrier. In such a way, we obtain a category $\katstructure(\metasig)$ of all ``available'' \metasig-structures. We are, however, free to choose only those structures, we are interested in.
\begin{definition}[Fifth parameter: Structures]
	\label{def:fifth-parameter}
	As the fifth parameter of a LFOC, we choose a certain subcategory $\katsemantics(\metasig)$ of the category $\katstructure(\metasig)$ of all \metasig-structures. 
\end{definition}
\begin{example}[FOL: Semantics]
	\label{ex:FOL:semantics}
	$\katsemantics(\metasig)=\katstructure(\metasig)$ comprises all \metasig-structures \structureU\ given by an arbitrary set \setU\ together with arbitrary subsets \(\semm{\featuremale}^\structureU\subseteq\katset(\{p\},\setU)=\setU^{\{p\}}\) and \(\semm{\featureparent}^\structureU\subseteq\katset(\{p_1,p_2,p_3\},\setU)=\setU^{\{p_1,p_2,p_3\}}\) as well as all homomorphisms between those \metasig-structures.
	
	Any ``terminological interpretation'' \(\calI\) in ALC includes the choice of a non-empty set \(\Delta^\calI\), called ``domain'', an interpretation of each ``concept name'' in \(N_C\) as a subset of \(\Delta^\calI\cong\katset(\{p\},\Delta^\calI)\) and an interpreation of each ``role name'' in \(N_R\) as a subset of \(\Delta^\calI\times\Delta^\calI\cong\katset(\{p_1,p_2\},\Delta^\calI)\). Obviously, this corresponds to our concept of structure.
\end{example}

\begin{example}[Category Theory: Semantics]
\label{ex:category-theory:semantics}
Analogously to Example \ref{ex:FOL:semantics},  $\katsemantics(\metasig)=\katstructure(\metasig)$ comprises all \metasig-structures \structureU\ given by an arbitrary graph \graphU\ together with arbitrary subsets \(\semm{\featureid}^\structureU\subseteq\katgraph(\arity(\featureid),\graphU)\) and \(\semm{\featurecomp}^\structureU\subseteq\katgraph(\arity(\featurecomp),\graphU)\).
That is, we include also structures like ``categories without identities'', ``categories with partial composition'' and so on.
Moreover, $\katsemantics(\metasig)$ includes all homomorphisms between those \metasig-structures. 
\end{example}
\begin{example}[Universal Algebra: Semantics]
\label{ex:universal-algebra:semantics}
$\katsemantics(\metasig)$ contains exactly one \metasig-structure $\structureU=(\carrierU,\setpred^\structureU)$ with \(\carrierU\isdef\graphofcatgen{\katset}\).
We choose only the singleton, containing the ``empty tuple'',  \[\semm{\featurefinal}^\structureU\isdef\{\interpretationa:\arity(\featurefinal)\to\graphofcatgen{\katset}\mid \interpretationa(pv)=\{()\}\,\}
\]
and use only Cartesian products \[\semm{\featureprodn}^\structureU\isdef\{\interpretationa:\arity(\featureprodn)\to\graphofcatgen{\katset}\mid \interpretationa(pv)=\interpretationa(pv_1)\times\ldots\times\interpretationa(pv_n)\}.\] 
Moreover, $\interpretationa:\arity(\featuretermtwo)\to\graphofcatgen{\katset}$ is in $\semm{\featuretermtwo}^\structureU$ iff \(\interpretationa(pv_3)=\interpretationa(pv_1)\times\interpretationa(pv_2)\) and \(\interpretationa(pe_4)=\la\interpretationa(pe_1),\interpretationa(pe_2)\ra;\interpretationa(pe_3)\).
\end{example}
\begin{open}[Multiple semantic universes in DPF?]
\label{open:multiple-semantic-universes-in-dpf}
We should check once if the idea to have more than one structure in $\katsemantics(\metasig)$ may work for DPF. We could, for example, consider the three categories \katset, \katpar\ and \katmult\ with the corresponding embedding functors.
\end{open}
\section{First-order Feature Expressions}
\label{sec:first-order-expressions}

\subsection{Syntax of Feature Expressions}
\label{sec:syntax-feature-expressions}

By a ``feature expression'' we mean something like a ``formula with free variables'' in traditional FOL. We do not consider them as formulas, but rather as ``derived features''. For us, a ``formula'' is, semantically seen, the subject of being ``valid or not valid'' in a given structure, while the semantics of a ``feature expression'', with respect to a given structure, is the ``set of all its ``solutions'', i.e., the set of all valid interpretations of the ``derived feature'' in this structure. We experience this perspective as the most adequate one, when formalizing and working with ``constraints'' in Modeling Driven Engineering. We have not seen ``closed formulas'' playing any role in this application area!  
\begin{definition}[Feature expressions: Syntax]
\label{def:feature-expressions:syntax}
For a footprint $\metasig=(\setpred,\arity)$ over \kattype\ we define inductively the set of all \textbf{(first-order) feature \metasig-expressions $\expression$ with arity \typeX},  $\typeX\arityof\expression$ in symbols, where \typeX\ varies, in parallel, over all the objects in \katvar: 
\begin{enumerate}
	\item Atomic expression: $\typeX\arityof\predP(\substitutiona)$ for any symbol $\predP\in\setpred$ and morphism $\substitutiona:\a\predP\to\typeX$ in \kattype.
	\item Everything: $\typeX\arityof\top$ for any object \typeX\ in \kattype.
	\item Void: $\typeX\arityof\bot$ for any object \typeX\ in \kattype.
	\item Conjunction: $\typeX\arityof(\expressiona\wedge\expressionb)$ for any expressions $\typeX\arityof\expressiona$ and $\typeX\arityof\expressionb$.
	\item Disjunction: $\typeX\arityof(\expressiona\vee\expressionb)$ for any expressions $\typeX\arityof\expressiona$ and $\typeX\arityof\expressionb$.
	\item Negation: $\typeX\arityof\neg\expression$ for any expression $\typeX\arityof\expression$.
	\item Conditional existential quantification:
	 $\typeX\arityof\condexistential{\expressiona}{\typemorphisma}{\typeY}{\expressionb}$ for any expressions $\typeX\arityof\expressiona$, $\typeY\arityof\expressionb$ and any morphism $\typemorphisma:\typeX\to\typeY$ in \kattype.
	\item Conditional universal quantification:   $\typeX\arityof\conduniversal{\expressiona}{\typemorphisma}{\typeY}{\expressionb}$  for any expressions $\typeX\arityof\expressiona$, $\typeY\arityof\expressionb$ and any morphism $\typemorphisma:\typeX\to\typeY$ in \kattype.
\end{enumerate}
\end{definition}

We included the case \textit{Void} after we became acquainted with Description Logics. 

\begin{remark}[Arity]
\label{rem:arity}
CALCO-reviewer 3 proposed to use the phrase "expressions with variables in \typeX" instead of "expressions with arity \typeX". We are not following this proposal for two reasons: (1) We can not talk about "variables in \typeX" since \typeX\ is just an object in an arbitrary base category. (2)  We consider feature expressions as "derived features" thus they should have an arity in the same way as the "basic features" in a footprint do have an arity. 
\end{remark}
\begin{remark}[Notation for expressions]
\label{rem:notation-expressions}
In traditional FOL, \typeX\ and \typeY\ are sets of variables and, instead of arbitrary maps $\typemorphisma:\typeX\to\typeY$, only inclusion maps $in_{\typeX,\typeY}:\typeX\hookrightarrow\typeY$ are considered in quantifications. Moreover, only the quantified variables  \(\typeY\setminus\typeX\) are recorded while \typeY\ has to be (re)constructed as the union \(\typeX\cup(\typeY\setminus\typeX)\). In other words, our \typeY\ lists all (!) variables that are allowed to appear as ``free variables'' in \expressionb! We record the whole \typeY\ for three reasons: (1) Even in arbitrary pre-sheaf topoi, we do not have complements. (2) We quantify actually over morphisms with source \typeY\ when we define the semantics of quantifications (compare Definition \ref{def:feature-expressions:semantics}) (3) In contrast to traditional FOL, $\typemorphisma:\typeX\to\typeY$ is allowed to be non-monic.

%

We allow non-monic morphims $\typemorphisma:\typeX\to\typeY$, even in  case \(\katvar\sqsubseteq\katset\), to express identifications without ``polluting'' our logic with equational symbols. Anyway, there seems to be no way to encode identifications by means of equational symbols if we go beyond presheaf topoi!?

In case $\typeX=\typeY$ and $\typemorphisma=id_\typeX$, both quantifications become obsolete and we get a kind of ``propositional implication'' that we can just write as  $\typeX\arityof\propimplication{\expressiona}{\expressionb}$.

In case $\expressiona=\top$, quantification becomes un-conditional, and we can just write $\typeX\arityof\existsvia{\typemorphisma}{\typeY}{\expressionb}$ and $\typeX\arityof\forallvia{\typemorphisma}{\typeY}{\expressionb}$ for the corresponding existential quantification $\typeX\arityof\condexistential{\top}{\typemorphisma}{\typeY}{\expressionb}$ and universal quantification $\typeX\arityof\conduniversal{\top}{\typemorphisma}{\typeY}{\expressionb}$, respectively.

If our base category \katbase\ is a pre-sheaf topos, like \katset\ or \katgraph, for example, inclusions of sets give us ``inclusion morphisms'' at hand. In case,  $\typemorphisma=in_{\typeX,\typeY}:\typeX\hookrightarrow\typeY$ is such an ``inclusion morphism'', we will drop the subscript $\typemorphisma$ (see Examples \ref{ex:FOL:expressions}, \ref{ex:category-theory:expressions} and \ref{ex:universal-algebra:expressions}).


\end{remark}

\begin{remark}[Everything and Void]
\label{rem:everything-and-void}
$\top$ and \(\bot\) are  not ``logical constants'', but  ``special  feature symbols'', inbuilt in any LFOC (analogously to the equation symbol in traditional Universal Algebra). To make this statement fully precise, we have to assume that \katbase\ has an initial object \initialobject\  that is included in \katvar. \initialobject\ is then the  arity of $\top$ and \(\bot\), while the fixed semantics for any carrier \carrierU\ is given by the two subsets of the singleton \(\katbase(\initialobject,\carrierU)=\{{!}_\carrierU:\initialobject\to\carrierU\}\), namely \(\semm{\bot}^\carrierU=\emptyset \) and \(\semm{\top}^\carrierU=\{{!}_\carrierU\}\). Consequently, we could use the same notation as for atomic expressions, namely $\typeX\arityof\top({!}_\typeX)$ and $\typeX\arityof\bot({!}_\typeX)$ where ${!}_\typeX:\initialobject\to\typeX$ is the unique initial morphism into \typeX. 
\end{remark}
\begin{remark}[Closed formulas]
\label{rem:closed-formulas}
In the presence of an initial object \initialobject, feature expressions $\initialobject\arityof\expression$  correspond to what is usually called a ``closed formula'' in traditional FOL. 
Note, that both \textit{Conditional quantifications} will generate ``closed formulas'' only in case \(\typeX=\initialobject\) where \(\substitutiona={!}_\typeY:\initialobject\to\typeY\) is the only choice for \substitutiona, in this case.
\end{remark}
\begin{remark}[Substitutions]
\label{rem:substitutions}
For any \(\typeX\in\katvar\) we consider the set $Feat(\metasig,\typeX)$ of all feature \metasig-expressions with fixed arity \typeX, i.e.,  $\expression\in Feat(\metasig,\typeX)$ if and only if the statement  $(\typeX\arityof\expression)$ can be derived with the rules in Definition \ref{def:feature-expressions:syntax}. 

Note, that the sets $Feat(\metasig,\typeX)$ are not disjoint. Each of the sets contains, for example, the special feature symbols \(\bot\) and \(\top\). Therefore, we work in Definition \ref{def:feature-expressions:syntax} with the notation $\typeX\arityof\;\ldots$.

It is desirable that any ``variable substitution'' $\typemorphismb:\typeX\to\typeZ$ in \kattype\ induces a ``substitution map'' from  $Feat(\metasig,\typeX)$ into $Feat(\metasig,\typeZ)$. Due to the two \textit{Conditional quantifications}, however, we can only define those maps if \katbase\ has pushouts for (certain) spans in \katvar\ and if we can choose pushouts,  in such a way, that \katvar\ is closed under them. Even then, the assignments \(\typeX\mapsto Feat(\metasig,\typeX)\) may be extended only to a pseudo-functor from \katvar\ into \katset, since chosen pushouts are, in general, not compositional. 

We do not need ``substitution maps'' for the general definition of LFOC's in this paper. They will be, however, essential for a LFOC to be useful in applications (compare the following  examples). 
A systematic investigation of the additional conditions for the categories \katbase\ and \kattype\ that ensure the availability of ``application-friendly substitution maps'' is another topic for future research.

If we extend once the framework by operations, the arrows \substitutiona\ in atomic expressions $\typeX\arityof\predP(\substitutiona)$ will be substitutions, i.e., Kleisli morphisms substituting variables by terms (compare \cite{WolterDK18}) Every Kleisli morphism should induce then also a map from  $Feat(\metasig,\typeX)$ into $Feat(\metasig,\typeZ)$!?
\end{remark}

\begin{example}[FOL: Expressions]
\label{ex:FOL:expressions}
If we consider \featuremale\ and \featureparent\ as traditional unary resp.\ tertiary predicate symbols, we can construct an ``open formula'' like \\
\hspace*{8em}\((\exists x_1 \exists x_2 \exists x_3:\featureparent(p,x_2,x_3)\wedge\featureparent(x_1,x_2,x_3))
\) \\
with \(p\) the only free variable. We consider this as a derived unary property ``being sibling of someone''. 
We could introduce an auxiliary symbol \featuresibling\ and use \(\featuresibling(p)\) as a shorthand (macro) for this open formula. The conjunction \((\featuremale(p)\wedge\featuresibling(p))\) would then represent a derived unary property ``being brother of someone''.


We consider the tuple \((p,x_2,x_3)\) just as a convenient notation of the  assignments \((p_1\mapsto p,p_2\mapsto x_2,p_3\mapsto x_3)\) defining a map \substitutiona\ from the arity \(\{p_1,p_2,p_3\}\) of \featureparent\ into the set \(\{p,x_1,x_2,x_3\}\) of variables (compare \cite{WolterDK18}). Relying on the conventions in Remark \ref{rem:notation-expressions}, we can then represent the derived property ``being sibling of someone'' by the feature expression \\
\hspace*{5em} \(\{p\}\arityof\existsvia{}{\{p,x_1,x_2,x_3\}}{(\featureparent(p,x_2,x_3)\wedge\featureparent(x_1,x_2,x_3))} \).

We would like to use this feature expression to define an auxiliary feature symbol \featuresibling\ with arity \(\{p\}\). To ensure, that then any feature expression $\typeX\arityof\expression$, containing \featuresibling, can be expanded into an equivalent feature expression $\typeX\arityof\expression'$, containing only the original feature symbols \featuremale\ and \featureparent, we need ``substitution maps'', as discussed in Remark \ref{rem:substitutions}.

A property ``having at most one couple of parents'' is represented by the feature expression\\
 \(\{p\}\hspace{0pt}\arityof\hspace{0pt}\forallvia{}{\{p,x_1,x_2,x_3,x_4\}}{\condexistential{(\featureparent(p,x_1,x_2)\wedge\featureparent(p,x_3,x_4))}{\substitutiona}{\{x_5,x_6,x_7\}}{\top}} \).\\
with \(\substitutiona:\{p,x_1,x_2,x_3,x_4\}\to\{x_5,x_6,x_7\}\) defined by \(p\mapsto x_5;x_1,x_3\mapsto x_6; x_2,x_4\mapsto x_7\).

ALC focuses on ``derived concepts'', i.e., feature expressions with \typeX\ a singleton. To describe, however, all ``derived concepts'' as feature expressions, we have to use arbitrary finite sets of variables.
To simulate the restriction on ``derived concepts'', we have to vary rules (7) and (8). Using our notational conventions, the ALC construct ``universal restriction \(\forall\featuregen{R}.\featuregen{C}\) for any role \(\featuregen{R}\in N_R\), any (derived) concept \featuregen{C}'', can be described as follows: For any role \(\featuregen{R}\) in \(N_R\), any expression \(\{p_1\}\arityof\featuregen{C}\) and any variables \(x_1,x_2\), not appearing in \featuregen{C}, we have:\;
\( \{x_1\}\arityof\forallvia{}{\{x_1,x_2\}}{\featuregen{R}(x_1,x_2)\to\featuregen{C}_\substitutiona(x_2)} \) 
where \(\substitutiona:\{p_1\}\to\{x_1,x_2\}\) is given by \(\substitutiona(p_1)=x_2\) and the expression \(\{x_1,x_2\}\arityof\featuregen{C}_\substitutiona(x_2)\) is obtained by substituting each occurrence of \(p_1\) in \featuregen{C} by \(x_2\) and by extending each variable declaration \typeY\ in \featuregen{C} by the ``fresh variable'' \(x_1\).
Analogously, the ALC construct ``the existential restriction \(\exists\featuregen{R}.\featuregen{C}\) of a concept \featuregen{C} by a role \(\featuregen{R}\in N_R\)'' can be described by existential quantification: For any role \(\featuregen{R}\) in \(N_R\), any expression \(\{p_1\}\arityof\featuregen{C}\) and any variables \(x_1,x_2\), not appearing in \featuregen{C}, we have:\;
\( \{x_1\}\arityof\existsvia{}{\{x_1,x_2\}}{\featuregen{R}(x_1,x_2)\wedge\featuregen{C}_\substitutiona(x_2)} \).

CALCO-reviewer 3 confirmed our translation of ACL into LFOC: The feature expressions that you get
for the ALC restrictions are precisely those that you get for
their translation to first-order logic along the standard
encoding of ALC in FOL. This is no surprise, but if you did this
independently, you can regard this as a confirmation of your
intuition.
\end{example}
\begin{example}[Category Theory: Expressions]
\label{ex:category-theory:expressions}
( We are sorry, put we will postpone a diagrammatic presentation of this example to the next version of the paper.)

The existence of an identity morphism for a certain single vertex can be stated by means of the feature expression \(\typeX\arityof\existsvia{}{\arity(\featureid)}{\featureid(pe)}\), where graph \typeX\ consists only of a vertex \(pv\). ``\((pe)\)'' is considered as a notation of the assignment \((pe\mapsto pe)\) representing uniquely the graph homomorphism \(id_{\arity(\featureid)}\). To express the uniqueness of the  identity morphism for a certain vertex, we do have the expression\\
\hspace*{6em}\(\typeX\arityof\forallvia{}{\typeY}{\condexistential{\featureid(pe_1)\wedge\featureid(pe_2)}{\substitutiona}{\arity(\featureid)}{\top}} \),\\
where \typeY\ is a graph with one node \(pv\) and two loops \(pe_1,pe_2\), and \(\substitutiona:\typeY\to\arity(\featureid)\) is given by the assignments \((pv\mapsto pv;pe_1,pe_2\mapsto pe)\). Analogously, we can express existence and uniqueness of composition of a certain pair of edges. At the end, we can express any properties, like monic, jointly monic, commutative (co)cone, existence and uniqueness of mediating morphisms, limit cone, colimit cone and so on, by means of feature expressions. 

We can also hide auxiliary items. The property ``commutative square'', e.g., is given by  \(\;\typeX\arityof\existsvia{}{\typeY}{\featurecomp(pe_1,pe_3,pe_5)\wedge\featurecomp(pe_2,pe_4,pe_5)}\), where the arity \typeX\ is a square with edges \(pv_1\stackrel{pe_1}{\to}pv_2\stackrel{pe_3}{\to}pv_4\) and \(pv_1\stackrel{pe_2}{\to}pv_3\stackrel{pe_4}{\to}pv_4\) and \typeY\ extends \typeX\ by an arrow \(pv_1\stackrel{pe_5}{\to}pv_4\). 
\end{example}
\begin{example}[Universal Algebra: Expressions]
\label{ex:universal-algebra:expressions}
In this case, feature expressions allow us to formulate (not fully precisely) properties like ``being a constant'', ``being a Cartesian product of'', ``being a variable (projection)'', ``being an operation/term''. For a more precise characterization, we have to work with ``typed arities'' as indicated in Example \ref{ex:universal-algebra:footprint}.
Analogously to the ``commutative square'' expression in Example \ref{ex:category-theory:expressions}, we can also hide ``subterms''.

Note, that we have the property ``two terms are equal'' at hand by means of the simple feature expression \(\typeX\arityof\existsvia{\substitutiona}{\typeY}{\top} \), with \typeX\ given by two parallel edges \(pe_1,pe_2:pv_1\to pv_2\) and \typeY\ consisting of one edge \(pe:pv_1\to pv_2\). \substitutiona\ is the only graph homomorphism from \typeX\ into \typeY. 
\end{example}

\subsection{Semantics of Feature Expressions}
\label{sec:semantics-of-feature-expressions}
Relying on Definition \ref{def:structures} of the semantics of feature symbols, we are going to define the semantics of a feature expression $\typeX\arityof\expression$ in a \metasig-structure  $\structureU =(\carrierU,\setpred^\structureU)$ as the set $\semm{\expression}_\typeX^\structureU$ of all \textbf{''solutions'' (valid interpretations)} of $\typeX\arityof\expression$ in \structureU. This semantics is a restriction of the semantics of \typeX\ relative to the carrier \carrierU, 
i.e., \(\semm{\expression}_\typeX^\structureU\subseteq\semm{\typeX}^\carrierU=\katbase(\typeX,\carrierU)\).

For \textbf{interpretations} $\interpretationa:\typeX\to\carrierU$, we will use, instead of $\interpretationa\in\semm{\expression}_\typeX^\structureU$, also the equivalent, but more traditional, notation \solution{\interpretationa}{\structureU}{\typeX}{\expression} .

Given a morphism \flar{\typemorphisma}{\typeX}{\typeY} in \kattype, we say that an interpretation \flar{\interpretationb}{\typeY}{\carrierU}  is an \textbf{extension} of an \textbf{interpretation} \flar{\interpretationa}{\typeX}{\carrierU} \textbf{via} \typemorphisma\ if, and only if, $\,\typemorphisma;\interpretationb=\interpretationa$.
$$\xymatrix{
	\typeX \ar[rr]^\typemorphisma\ar[dr]_(.4)\interpretationa
	& {}\ar@{}[d]|(.4)=
	& \typeY  \ar[dl]^(.4)\interpretationb
	\\
	& \carrierU
}$$
\begin{definition}[Feature expressions: Semantics]
\label{def:feature-expressions:semantics}
The semantics of feature \metasig-expressions in an arbitrary, but fixed, \metasig-structure  $\structureU =(\carrierU,\setpred^\structureU)$ is defined inductively:
\begin{enumerate}
\item Atomic expressions: $\interpretationa\in\semm{\predP(\substitutiona)}_\typeX^\structureU$ \; iff \;$\substitutiona;\interpretationa\in\semm{\predP}^\structureU$ 
$$\xymatrix{
\a\predP \ar[rr]^\substitutiona\ar[dr]_(.4){\substitutiona;\interpretationa}
& {}\ar@{}[d]|(.4)=
& \typeX  \ar[dl]^(.4)\interpretationa
\\
& \carrierU
}$$
\item Everything: $\semm{\top}_\typeX^\structureU\isdef\semm{\typeX}^\carrierU=\katbase(\typeX,\carrierU)$
\item Void: $\semm{\bot}_\typeX^\structureU\isdef\emptyset$
\item Conjunction: $\semm{\expressiona\wedge\expressionb}_\typeX^\structureU
\isdef\semm{\expressiona}_\typeX^\structureU\cap\semm{\expressionb}_\typeX^\structureU$
\item Disjunction:  $\semm{\expressiona\vee\expressionb}_\typeX^\structureU
\isdef\semm{\expressiona}_\typeX^\structureU\cup\semm{\expressionb}_\typeX^\structureU$
\item Negation:  $\semm{\neg\expression}_\typeX^\structureU
\isdef\katbase(\typeX,\carrierU)\setminus\semm{\expression}_\typeX^\structureU$
\item Conditional existential quantification: $\interpretationa\in\semm{\condexistential{\expressiona}{\typemorphisma}{\typeY}{\expressionb}}_\typeX^\structureU$  \; iff \; \\ $\interpretationa\in\semm{\expressiona}_\typeX^\structureU$ implies that there exists an extension \interpretationb\ of \interpretationa\ via \typemorphisma\ such that $\interpretationb\in\semm{\expressionb}_\typeY^\structureU$.
$$\xymatrix{
\typeX \ar[rr]^\substitutiona\ar[dr]_(.3){\solution{\interpretationa}{\structureU}{\typeX}{\expressiona}}
& {}\ar@{}[d]|(.4)=
& \typeY  \ar[dl]^(.3){\exists\solution{\interpretationb}{\structureU}{\typeY}{\expressionb}}
\\
& \carrierU
}$$
\item Conditional universal quantification:  $\interpretationa\in\semm{\conduniversal{\expressiona}{\typemorphisma}{\typeY}{\expressionb}}_\typeX^\structureU$  \; iff \; \\ $\interpretationa\in\semm{\expressiona}_\typeX^\structureU$ implies that for all extensions \interpretationb\ of \interpretationa\ via \typemorphisma\ it holds that $\interpretationb\in\semm{\expressionb}_\typeY^\structureU$.
$$\xymatrix{
\typeX \ar[rr]^\substitutiona\ar[dr]_(.3){\solution{\interpretationa}{\structureU}{\typeX}{\expressiona}}
& {}\ar@{}[d]|(.4)=
& \typeY  \ar[dl]^(.3){\forall\solution{\interpretationb}{\structureU}{\typeY}{\expressionb}}
\\
& \carrierU
}$$
\end{enumerate}
\end{definition}

\begin{remark}[Feature expressions: Semantics]
\label{rem:feature-expressions:semantics}
Every feature symbol \predP\ in \setpred\ reappears as the  \metasig-expression $\a\predP\arityof\predP(id_{\a \predP})$ and Definition \ref{def:feature-expressions:semantics} ensures $\semm{\predP(id_{\a\predP})}_{\a\predP}^\structureU=\semm{\predP}^\structureU$.

As usual, any conditional quantification becomes valid, if the condition does not hold, i.e., if 
$\interpretationa\notin\semm{\expressiona}_\typeX^\structureU$. This gives us ``guarded constraints'' at hand that need only to be checked if  the ``guard'' \expressiona\ becomes valid. This is quite relevant for practical applications.

If 
$\interpretationa\in\semm{\expressiona}_\typeX^\structureU$, the conditional universal quantification $\typeX\arityof\conduniversal{\expressiona}{\typemorphisma}{\typeY}{\expressionb}$ is trivially valid if there is no extension \interpretationb\ of \interpretationa\ at all, while the corresponding conditional existential quantification is not valid, in this case. 
\end{remark}

\begin{open}[Constructive feature expressions]
\label{open:constructive-feature-expressions}
A feature expression $\typeX\arityof\expressiona$ may be called \textbf{constructive} if it does not contain Negation or Conditional universal quantification, respectively. Our conjecture is that the semantics of constructive predicate expressions is preserved by homomorphisms between \metasig-structures.
\end{open}

\section{Institutions of First-Order Constraints}
\label{sec:institutions-of-contraints}
Generalizing concepts like ``set of generators'' in Group Theory, ``underlying graph of a sketch'' in Category Theory, ``set of individual names'' in Description Logics and ``underlying graph of a model'' in Software Engineering, we coin in this section the concept ``context''. Further, we introduce ``constraints'' in generalizing the corresponding concepts ``defining relation'' in Group Theory, ``diagram in a sketch'' in Category Theory, ``concept/role assertion'' in Description Logic and ``constraint'' in Software Engineering.  We use the concept of institution  \cite{Dia08,GB92} as a methodological guideline to develop and present corresponding LFOC's.

\subsection{Category of Contexts and Sentence Functor}
\label{sec:contexts-and-sentence-functor}
As \textbf{''signatures''}, in the sense of institutions, we introduce ``contexts''.
\begin{definition}[Sixth parameter: Contexts]
	\label{def:sixth-parameter}
	As the sixth parameter of a LFOC,  we choose a subcategory \katcontext\ of \katbase\ such that \kattype\ is a subcategory of \katcontext. The objects in \katcontext\ are called  \textbf{''contexts''} while we refer to the morphisms in \katcontext\ as \textbf{''context morphisms''}. 
\end{definition}
\begin{remark}[Variables vs. context]
\label{rem:variables-vs-context}
Introducing ``contexts'', we establish a technological layer between ``pure syntax'' (variable declarations) and ``pure semantics'' (carriers of structures). We may use the term \textbf{pseudo syntax}. We prefer to consider variable declarations as something finite or enumerable while contexts can be arbitrary. If we are interested in completeness proofs and corresponding free structures, for example, we may have contexts that are or become carriers of structures. Or, to say it in the light of Remark \ref{rem:signature-extensions-vs-generators}: We refuse ad-hoc extensions of footprints (signatures) by carriers of structures.

We perceive the inclusion $\kattype\sqsubseteq\katcontext$ as a ``change of roles'': Variables are essentially syntactic items but can also serve as ``generators of structures'', like groups and (term) algebras, for example.
\qed
\end{remark}
As \textbf{''sentences''}, in the sense of institutions, we consider ``constraints''.
\begin{definition}[Constraint]
\label{def:constraint}
A \textbf{\metasig-constraint} \constraint{\typeX}{\expression}{\bindinga} on a context $\contextK\in\katcontext$ is given by a feature \metasig-expression $\typeX\arityof\expression$ and a \textbf{binding morphism} $\flar{\bindinga}{\typeX}{\contextK}$ in \katcontext. 

By $\functorCtr(\contextK)$ we denote the set of all \metasig-constraints on \contextK.
\end{definition}
\begin{remark}[Expression vs. constraint]
\label{rem:expression-vs-constraint}
We distinguish between expressions and constraints for, at least, three reasons: (1) This corresponds to the situation in sketches where we distinguish between a property and its arity, on one side, and a diagram, on the other side. (2) We consider expressions as pure syntactic entities. (3) Using constraints, we can encapsulate the construction of syntactic entities and can realize ``changes of bases'' in any category \katbase\ by simple composition. This trick we adapt from  \cite{GB92} where it is used for ``initial/free constraints''.
\qed
\end{remark}

Any morphism \flar{\contextmorphisma}{\contextK}{\contextG} in \katcontext\ induces a map \flar{\functorCtr(\contextmorphisma)}{\functorCtr(\contextK)}{\functorCtr(\contextG)} defined by simple post-composition: $\functorCtr(\contextmorphisma)\constraint{\typeX}{\expression}{\bindinga}\isdef\constraint{\typeX}{\expression}{\bindinga;\contextmorphisma}$ for all constraints \constraint{\typeX}{\expression}{\bindinga} on \contextK.
$$\xymatrix{
&  \typeX\arityof\expression\ar[dl]_(.6){\bindinga}^(.4){}="p1"\ar[dr]^(.6){\bindinga;\contextmorphisma}_(.4){}="p2"
\\
\contextK\ar[rr]_\contextmorphisma
&& \contextG
\ar@/_1.2pc/@{|->} "p1";"p2"
}$$

It is easy to show, that the assignments $\contextK\mapsto\functorCtr(\contextK)$ and $\contextmorphisma\mapsto\functorCtr(\contextmorphisma)$ provide a functor \flar{\functorCtr}{\katcontext}{\katset}. This is our \textbf{''sentence functor''}, in the sense of institutions.

\subsection{Model functor}
\label{sec:model-functor}
\textbf{''Models''}, in the sense of institutions, are interpretations of contexts in structures.
\begin{definition}[Context interpretations]
\label{def:context-interpretations}
An \textbf{interpretation \contextinterpretation{\interpretationa}{\structureU} of a context} $\contextK\in \katcontext$ is given by a \metasig-structure $\structureU=(\carrierU,\setpred^\structureU)$ in $\katsemantics(\metasig)$ and a morphism \flar{\interpretationa}{\contextK}{\carrierU} in \katbase.

A morphism \flar{\varsigma}{\contextinterpretation{\interpretationa}{\structureU}}{\contextinterpretation{\interpretationb}{\structureV}} between two interpretations of \contextK\ is given by a morphism \flar{\varsigma}{\structureU}{\structureV} in $\katsemantics(\metasig)$ such that $\interpretationa;\varsigma=\interpretationb$ for the underlying morphism \flar{\varsigma}{\carrierU}{\carrierV} in \katcarrier.
$$\xymatrix{
	&  \contextK\ar[dl]_{\interpretationa}\ar[dr]^{\interpretationb}\ar@{}[d]|(.5)=
	\\
	\carrierU\ar[rr]^\varsigma
	&& \carrierV
}$$
For any context \contextK\ in \katcontext\ we denote by $\functorInt(\contextK)$ the category of all interpretations of \contextK\ and all morphisms between them. 
\end{definition}
Note, that $\functorInt(\contextK)$ is a discrete category (set) in case $\katsemantics(\metasig)$ consists of exactly one object and its identity morphism, as in case of our Universal Algebra example (see Example \ref{ex:universal-algebra:semantics}).
\begin{open}[Natural transformations?]
\label{open:natural-transfromations}
We present in this paper an abstract and general definition of LFOC's that fits all the applications, we do have in mind. Therefore we are not assuming any structure on the hom-sets \(\katbase(\contextK,\carrierU)\). 

There are, however, cases like the Universal Algebra example, for instance, where  \(\katbase(\contextK,\carrierU)\) is actually a category with natural transformations as morphisms. For those special cases we can vary Definition \ref{def:context-interpretations} in such a way that a morphism between the two interpretations of \contextK\ is given by a morphism \flar{\varsigma}{\carrierU}{\carrierV} in \katcarrier\ and a natural transformation in \(\katbase(\contextK,\carrierV)\) from  $\interpretationa;\varsigma$ to $\interpretationb$.  We are convinced that all the following constructions and results can be transferred, more or less, straightforwardly to this extended version of morphisms between interpretations. For the moment, we have, however, to let this extension for special LFOC's open as another topic of future research.   
\qed
\end{open}

Any context morphism \flar{\contextmorphisma}{\contextK}{\contextG} induces a functor \flar{\functorInt(\contextmorphisma)}{\functorInt(\contextG)}{\functorInt(\contextK)} defined by simple pre-composition: $\functorInt(\contextmorphisma)\contextinterpretation{\interpretationb}{\structureV}\isdef\contextinterpretation{\contextmorphisma;\interpretationb}{\structureV}$ for all interpretations \contextinterpretation{\interpretationb}{\structureV} of \contextG, and for any morphism \flar{\varsigma}{\contextinterpretation{\interpretationa}{\structureU}}{\contextinterpretation{\interpretationb}{\structureV}} between two interpretations of \contextG\ the same underlying  morphism \flar{\varsigma}{\carrierU}{\carrierV} in \katcarrier\ establishes a morphism \flar{\functorInt(\contextmorphisma)(\varsigma)\isdef\varsigma}{\contextinterpretation{\contextmorphisma;\interpretationa}{\structureU}}{\contextinterpretation{\contextmorphisma;\interpretationb}{\structureV}} between the corresponding two interpretations of \contextK.
$$\xymatrix{
	\contextK\ar[rr]^\contextmorphisma\ar[dr]_{\contextmorphisma;\interpretationb}^(.5){}="p2"
&& \contextG\ar[dl]^{\interpretationb}_(.5){}="p1"
&& \contextK\ar[rr]^\contextmorphisma\ar[d]_{\contextmorphisma;\interpretationa}\ar[drr]_(.35){\contextmorphisma;\interpretationb}
&& \contextG\ar[d]^{\interpretationb}\ar[dll]^(.35){\interpretationa}
	\\
&  \carrierV
&&&\carrierU\ar[rr]^\varsigma
&&  \carrierV
\ar@/_.8pc/@{|->} "p1";"p2"
}$$
It is straightforward to validate, that the assignments $\contextK\mapsto\functorInt(\contextK)$ and $\contextmorphisma\mapsto\functorInt(\contextmorphisma)$ define a functor \flar{\functorInt}{\katcontext^{op}}{\katcat}. This is our \textbf{''model functor''}, in the sense of institutions.

\subsection{Satisfaction Relation and Satisfaction Condition}
\label{sec:satisfaction-relation-and-satisfaction-condition}
The last two steps, in establishing an institution, are the definition of satisfaction relations and the proof of the so-called satisfaction condition. The satisfaction relations are simply given by the semantics of features expressions, as described in Definition \ref{def:feature-expressions:semantics}, and the composition of morphisms in \katbase.
\begin{definition}[Satisfaction relation]
\label{def:satisfaction-relation}
For any context $\contextK\in\katcontext$, any constraint \constraint{\typeX}{\expression}{\bindinga} on \contextK\ and any interpretation \contextinterpretation{\interpretationa}{\structureU} of context \contextK\ we define:
\begin{equation}\label{equ:satisfaction-relation}
\contextinterpretation{\interpretationa}{\structureU} \models_\contextK \constraint{\typeX}{\expression}{\bindinga} \quad \mbox{iff} \quad \solution{\bindinga;\interpretationa}{\structureU}{\typeX}{\expression} \quad (\,\mbox{i.e.} \quad
 \bindinga;\interpretationa\in\semm{\expression}_\typeX^\structureU\;)
\end{equation}
$$\xymatrix{
	& 
	\contextK\ar[dl]_{\interpretationa}
	\\   
	\carrierU
	&&
	\typeX\arityof\expression\ar[ul]_{\bindinga}\ar[ll]^{\bindinga;\interpretationa}
}$$
\end{definition}
Having developed everything in a systematic modular way, satisfaction condition is ``for free''.
\begin{corollary}[Satisfaction condition]
\label{coro:satisfaction-condition}
For any morphism \flar{\contextmorphisma}{\contextK}{\contextG} in \katcontext, any constraint \constraint{\typeX}{\expression}{\bindinga} on \contextK, and any interpretation  \contextinterpretation{\interpretationb}{\structureU} of context \contextG\ we have:
\begin{equation}\label{equ:satisfaction-condition}
\functorInt(\contextmorphisma)\contextinterpretation{\interpretationb}{\structureU}
\models_\contextK \constraint{\typeX}{\expression}{\bindinga}  \quad \mbox{iff} \quad \contextinterpretation{\interpretationb}{\structureU}
\models_\contextG\functorCtr(\contextmorphisma)\constraint{\typeX}{\expression}{\bindinga}.
\end{equation}
$$\xymatrix{
\contextK\ar[dd]_\contextmorphisma
& \contextinterpretation{\contextmorphisma;\interpretationb}{\structureU}\ar@{}[r]|(.45){\models_\contextK}
& \constraint{\typeX}{\expression}{\bindinga}\ar@{|->}[dd]^{\functorCtr(\contextmorphisma)}
&&
\contextK\ar[dd]^\contextmorphisma\ar[dl]_{\contextmorphisma;\interpretationb}
\\
&&& 
\carrierU 
&&
\typeX\arityof\expression\ar[ul]_{\bindinga}\ar[dl]^{\bindinga;\contextmorphisma}
\\
\contextG
& \contextinterpretation{\interpretationb}{\structureU}\ar@{}[r]|(.45){\models_\contextG}
\ar@{|->}[uu]^{\functorInt(\contextmorphisma)}
& 
\constraint{\typeX}{\expression}{\bindinga;\contextmorphisma}
&&
\contextG\ar[ul]^{\interpretationb}
}$$
\end{corollary}
\begin{proof}
Due to the definition of the functors \flar{\functorInt}{\katcontext^{op}}{\katcat} and \flar{\functorCtr}{\katcontext}{\katset}, we obtain the commutative diagram, above on the right, thus condition \ref{equ:satisfaction-condition} follows immediately from Definition \ref{def:satisfaction-relation}.
\end{proof}
\begin{remark}[Satisfaction Condition]
\label{rem:satisfaction-condition}
There is a certain similarity between the nearly trivial proof of the satisfaction condition in Corollary \ref{coro:satisfaction-condition} and the systematic and detailed proofs of the satisfaction condition for four formalisms in \cite{WKWC94}. Introducing in \cite{WKWC94} the concept of ``corresponding assignments'', those proofs became uniformly easy and straightforward.
\end{remark}

Summarizing all our definitions and results, we can state the main result of the paper:
\begin{theorem}[Institution of Constraints]
\label{theo:institution}
 Any choice of a base category \katbase, of subcategories \kattype, \katcontext, \katcarrier\ of \katbase\ such that $\kattype\sqsubseteq\katcontext$, of a footprint \metasig\ and of a category $\katsemantics(\metasig)$ of \metasig-structures establishes a corresponding \textbf{institution of constraints} $\mathcal{IC}=(\katcontext,\functorCtr,\functorInt,\models)$.
\end{theorem}

An institution of constraints $\mathcal{IC}=(\katcontext,\functorCtr,\functorInt,\models)$ is the essential building block of a LFOC. We are, however, still missing  expressive tools to axiomatize the semantics of features and to define the corresponding syntactic counterpart, namely deduction rules. 
We will outline this dimension of LFOC's in Section \ref{sec:sketch-rules}.

\section{Sketches and Sketch Morphisms}
\label{sec:sketches}

Any institution gives us a corresponding category of presentations and an extension of the ``model functor'' of the institution to the category of presentations \cite{Dia08,GB92}. 
We will outline this construction for institutions of constraints using, however, the term ``sketch'' instead of ``presentation''. 
\begin{definition}[Sketch]
\label{def:sketch}
A \textbf{(first-order) \metasig-sketch} $\sketchK=(\contextK,\setofconstraintsK)$ is given by a context $\contextK\in\katcontext$ and a set \setofconstraintsK\ of \, \metasig-constraints on \contextK.
\end{definition}
\begin{example}[FOL: Sketches]
\label{ex:FOL:sketches}
''Elementary diagrams'' are discussed in Remark \ref{rem:sketch-vs-structure}.

In ALC we meet contexts as ``sets \(N_O\) of individual names''. ``Assertional axioms'' are ``concept assertions'' or ``role assertions''. A ``concept assertion'', i.e., a statement of the form \(a:\featuregen{C}\) where \(a\in N_O\) and \featuregen{C} is a (derived) concept, can be seen as a constraint \constraint{\{p_1\}}{\featuregen{C}(p_1)}{(a)} on \(N_O\) where ``\((a)\)'' is just a convenient notation for a binding \(\bindinga:\{p_1\}\to N_O\) with \(\bindinga(p_1)=a\).
A ``role assertion'', i.e., a statement of the form \((a,b):\featuregen{R}\) where \(a,b\in N_O\) and \featuregen{R} is a role, can be seen as a constraint \constraint{\{p_1,p_2\}}{\featuregen{R}(p_1,p_2)}{(a,b)} on \(N_O\). An ``ABox'' is a finite set of ``assertional axioms''. So, a pair \((N_O,\calA)\) of a ``set \(N_O\) of individual names'' and an ABox \calA\ of ``assertional axioms'' on \(N_O\) is just a sketch in our sense.
\end{example}
\begin{example}[Category Theory: Sketches]
\label{ex:category-theory:sketches}
These are just the sketches, as we know them, with the essential difference, that we are not restricting ourselves to commutative, limit and colimit diagrams. We do not need to code, for example,  the property ``monic'' by means of a pullback, but can define it directly as a property of edges. We can define ``jointly monic'' directly as a property of a span without referring to products. We can require that only certain paths in a diagram commute and so on.
\end{example}
\begin{example}[Universal Algebra: Sketches]
\label{ex:universal-algebra:sketches}
On this abstraction level, we can represent any many-sorted algebraic signature \(\Sigma\)  by a sketch. Any set of \(\Sigma\)-terms, any set of \(\Sigma\)-equations, any equational \(\Sigma\)-specification can be represented by a sketch. A point is that ``categorical products'' are not needed to describe terms and equations, but only to reason about them. The fact, that Cartesian products are ``categorical products'' in \katset, can be added a posteriori by means of ``universal rules'' (see Remark \ref{rem:universal-rules-axioms}).
\qed
\end{example}
For any context $\contextK\in\katcontext$, any set $\setC\subseteq\functorCtr(\contextK)$ of constraints on \contextK\ and any interpretation \contextinterpretation{\interpretationa}{\structureU} of context \contextK\ in a \metasig-structure \structureU\ in $\katsemantics(\metasig)$ we define (see Definition \ref{def:satisfaction-relation}):
\begin{equation}\label{equ:satisfaction-set-of-constraints}
\contextinterpretation{\interpretationa}{\structureU} \models_\contextK \setC\quad \mbox{iff} \quad \contextinterpretation{\interpretationa}{\structureU} \models_\contextK \constraint{\typeX}{\expression}{\bindinga} \;\; \mbox{for all}\; \constraint{\typeX}{\expression}{\bindinga} \in\setC.
\end{equation}
Note, that the constraints in \setC\ may have different ``variable declarations''  \typeX.

\begin{definition}[Model of sketch]
\label{def:model-of-sketch}
A \textbf{model} of a \metasig-sketch $\sketchK=(\contextK,\setofconstraintsK)$ is an interpretation  \contextinterpretation{\interpretationa}{\structureU} of a context \contextK\ such that $\contextinterpretation{\interpretationa}{\structureU} \models_\contextK \setofconstraintsK$. We denote by $\functorMod(\sketchK)$ the full subcategory of $\functorInt(\contextK)$ determined by all models of \sketchK.
\end{definition}
To define reasonable morphisms between sketches, we have to take into account semantical entailments between constraints:  

\begin{definition}[Constraint entailment]
	\label{def:constraint-entailment}
	For any context $\contextK\in\katcontext$, any sets $\setC, \setD\subseteq\functorCtr(\contextK)$ of constraints on \contextK, we say that \setC\ \textbf{entails \(\setD\) (semantically)}, \(\setC\Vdash_\contextK\setD\) in symbols,  if, and only if, for all interpretations \contextinterpretation{\interpretationa}{\structureU} of \contextK: $\contextinterpretation{\interpretationa}{\structureU} \models_\contextK \setC\,$ implies $\,\contextinterpretation{\interpretationa}{\structureU} \models_\contextK \setD$.
\end{definition}

Note, that we quantify over all interpretations in all structures at once. This is equivalent to the usual two step approach: First, a quantification over all interpretations in single structures and, second, a quantification over all structures.

\begin{definition}[Sketch morphism]
\label{def:sketch-morphism}
A \textbf{morphism} $\contextmorphisma:\sketchK\to\sketchG$ between two \metasig-sketches $\sketchK=(\contextK,\setofconstraintsK)$ and $\sketchG=(\contextG,\setofconstraintsG)$ is given by a morphism $\contextmorphisma:\contextK\to\contextG$ in \katcontext\ such that $\setofconstraintsG\Vdash_\contextG \functorCstr(\contextmorphisma)(\setofconstraintsK)$. 
By $\katsketch(\metasig)$\ we denote the category of all \metasig-sketches. 
\end{definition}
Sketch morphisms $\contextmorphisma:\sketchK\to\sketchG$ with \(\contextK=\contextG\) and \(\contextmorphisma=id_\contextK\) correspond to constraint entailments.
\begin{remark}[Sketch vs. structure]
\label{rem:sketch-vs-structure}
In practical applications, we work only with finite contexts. If we are, however, interested in completeness proofs and/or the construction of free structures, for example, we have to assume that \katcarrier\ is a subcategory of \katcontext. In this case, we have two canonical ways to transform  a \metasig-structure $\structureU=(\carrierU,\setpred^\structureU)$ into a ``semantical \metasig-sketch''. The minimalist variant encodes only the semantics of feature symbols, that is, we take the \metasig-sketch $ \sketchS^\structureU_\setpred= (\carrierU,C^\structureU_\setpred)$
with \(C^\structureU_\setpred=\{\constraint{\a\predP}{\predP(id_{\a\predP})}{\interpretationa}\mid
\predP\in\setpred,\interpretationa\in\semm{\predP}^\structureU\}\). The maximal variant encodes the semantics of all feature expressions, that is, we take the \metasig-sketch $ \sketchS^\structureU= (\carrierU,C^\structureU)$ with  \(C^\structureU=\{\constraint{\typeX}{\expression}{\interpretationa}\mid
\typeX\in\katvar,\,\interpretationa\in\semm{\expression}^\structureU_\typeX\}\).
By construction, \((id_\carrierU,\structureU)\) is a model as well of \(\sketchS^\structureU_\setpred\) as of \(\sketchS^\structureU\) that is, moreover, initial in $\functorMod(\sketchS^\structureU_\setpred)$.  It can be shown, that the minimalist variant gives rise to a full embedding of \(\katsemantics(\metasig)\) into $\katsketch(\metasig)$.

There are no structures in \cite{DW08}. Instead, we worked, essentially, with the minimalist variant of ``semantical sketches''.  
In traditional FOL, we meet the maximal variant in form of \textbf{''elementary diagrams''} \cite{ChangKeisler1990}. The difference is that the carrier of a first-order structure is not taken as a ``context'', located between footprints (signatures) and structures. Instead, each element of the carrier is added as a constant to the signature. As long as \katbase\ is a pre-sheaf topos, it may be possible to use the same trick to produce ``syntactic encodings'' of structures. But, even so, we are convinced, that sketches provide a more direct and adequate tool to produce and to work with ``syntactic encodings'' of structures. 
The ``elementary diagram'' approach needs to encode homomorphisms between structures by means of signature morphisms. This looks inadequate since different abstraction levels are mixed up. 
\qed
\end{remark}

Since, we realize ``changes of bases'' by simple post-composition, we can construct a pushout in  $\katsketch(\metasig)$ by means of a pushout in \katbase\ and a union of translated constraints.
\begin{corollary}[Pushouts]
\label{coro:pushouts}
$\katsketch(\metasig)$ has pushouts as long as \katbase\ has pushouts.
\end{corollary}

For any sketch morphism  $\contextmorphisma:\sketchK\to\sketchG$ the condition $\setofconstraintsG\Vdash_\contextG \functorCstr(\contextmorphisma)(\setofconstraintsK)$ ensures, due to the satisfaction condition, that the functor \flar{\functorInt(\contextmorphisma)}{\functorInt(\contextG)}{\functorInt(\contextK)} restricts to a functor from \functorMod(\sketchG) into \functorMod(\sketchK). In such a way, the assignments $\sketchK\mapsto\katmodels(\sketchK)$ extend to a functor $\functorgen{Mod}:\katsketch^{op}\to\katcat$.
\begin{remark}[Amalgamation]
\label{rem:amalgamation}
Since pushouts in $\katsketch(\metasig)$ are based on pushouts in \katbase, we get, trivially, what is called ``amalgamation'' in Algebraic Specifications \cite{EM85}. Abstractly formulated: \katbase\ having pushouts ensures also that $\functorgen{Mod}:\katsketch(\metasig)^{op}\to\katcat$ is continuous, that is, maps pushouts in  $\katsketch(\metasig)$ into pullbacks in \katcat\ \cite{EGRW98,WK12_JACS}.
\end{remark}

\begin{remark}[Software models]
\label{rem:software-models}
In Software Engineering, sketches appear plainly as  appropriate formalizations of a broad variety of Software Models. In our MDSE papers, we call them, however, specifications (or models) since we experienced that engineers can not perceive a ``sketch'' as something with precise syntax and semantics.
The category  $\katsketch(\metasig)$ is the tool of choice to describe and to study relations between software models and different ways to construct new models out of given ones  (often by pushout constructions).
\qed
\end{remark}

Constraint entailments describe properties of the chosen semantics of our logic. The other way around, we can use them to formulate requirements for the intended semantics of the feature symbols in \metasig. Since constraint entailments concern only constraints on a fixed context, they are, however, not expressive enough to axiomatize, for example, that all (!) vertices do have an identity. To express those kinds of requirements, we need an appropriate variant of the ``sketch-entailments'' in \cite{Mak97}. We call them ``sketch rules''.

\section{Sketch Rules}
\label{sec:sketch-rules}
We get sketch morphisms in a canonical way, for any institution of constraints. We isolate the ``syntactic part'' of sketch morphisms as a concept of its own.
\begin{definition}[Sketch rule]
\label{def:sketch-rule}
A \textbf{\metasig-sketch rule} $\sketchL\stackrel{\contextmorphisma}{\Rightarrow}\sketchR$ is given by two \metasig-sketches $\sketchL=(\contextL,\setofconstraintsL)$, $\sketchR=(\contextR,\setofconstraintsR)$ and a context morphism $\contextmorphisma:\contextL\to\contextR$.
\end{definition}

\begin{definition}[Conservation]
\label{def:conservation}
A \metasig-structure $\structureU=(\carrierU,\setpred^\structureU)$ is \textbf{conservative (model-expansive) w.r.t.\ a \metasig--sketch rule} $\sketchL\stackrel{\contextmorphisma}{\Rightarrow}\sketchR$ if, and only if, each model \contextinterpretation{\interpretationa}{\structureU} of \sketchL\ in \structureU, i.e., \(\contextinterpretation{\interpretationa}{\structureU} \models_\contextL \setofconstraintsL\), can be extended to a model \contextinterpretation{\interpretationb}{\structureU} of \sketchR\ in \structureU, i.e., \(\contextmorphisma;\interpretationb=\interpretationa\) and \(\contextinterpretation{\interpretationb}{\structureU} \models_\contextR \setofconstraintsR\).
$$\xymatrix{
	\contextL \ar[rr]^\substitutiona\ar[dr]_(.3){\contextinterpretation{\interpretationa}{\structureU} \models_\contextL\setofconstraintsL}
	& {}\ar@{}[d]|(.4)=
	& \contextR  \ar[dl]^(.3){\exists\contextinterpretation{\interpretationb}{\structureU} \models_\contextR \setofconstraintsR}
	\\
	& \carrierU
}$$
$\sketchL\stackrel{\contextmorphisma}{\Rightarrow}\sketchR$ is \textbf{sound (in $\mathcal{IC}$)} if, and only if,  all \metasig-structures \structureU\ in $\katsemantics(\metasig)$ are conservatice w.r.t.\ $\sketchL\stackrel{\contextmorphisma}{\Rightarrow}\sketchR$.
\end{definition}
In case \(\contextL=\contextR\) and \(\contextmorphisma=id_\contextL\), we write just $\sketchL\Rightarrow\sketchR$ instead of $\sketchL\stackrel{\contextmorphisma}{\Rightarrow}\sketchR$. 
Soundness reduces, in this case, simply to constraint entailment: Each model \(\contextinterpretation{\interpretationa}{\structureU}\) of \sketchL\ is also a model of \sketchR, i.e., \(\contextinterpretation{\interpretationa}{\structureU} \models_\contextL \setofconstraintsL\) implies \(\contextinterpretation{\interpretationa}{\structureU} \models_\contextL \setofconstraintsR\).
\begin{remark}[Sketch rule vs. sketch morphism]
\label{rem:sketch-rule-vs-sketch-morphism}
Besides the overlap on constraint entailments, the concepts ``sketch morphism'' and ``sound sketch rule'' are skewed. For non-isomorphic \contextmorphisma, sketch morphisms talk about ``reducts of models'' while ``sound sketch rules'' state the existence of ``model extensions''. 

The fact, that the satisfaction condition ensures that each sound rule $\sketchL\stackrel{\contextmorphisma}{\Rightarrow}\sketchR$ gives rise to a trivial sketch morphism $\contextmorphisma:\sketchL\to\sketchR^\contextmorphisma$ with \(\sketchR^\contextmorphisma=(\contextR,\setofconstraintsR\cup\functorCstr(\contextmorphisma)(\setofconstraintsL))\), only highlights that our ``extensions'' are actually ``persistent extensions''. 
\end{remark}
\begin{remark}[Universal rules and Axioms]
\label{rem:universal-rules-axioms}
There are \textbf{universal sketch rules} (or, more precisely, rule schemata) that are sound in any LFOC, since they reflect the structure and semantics of feature expressions. Especially, they may describe the ``folding'' and ``unfolding'' of feature expressions. In case of ``conjunction'', for example, we do have the two sketches \(\sketchL=(\typeX,\{(\typeX\arityof\expressiona\wedge\expressionb,id_\typeX)\})\) and \(\sketchR=(\typeX,\{(\typeX\arityof\expressiona,id_\typeX),(\typeX\arityof\expressionb,id_\typeX)\})\). As well the ``unfolding'' rule $\sketchL\Rightarrow\sketchR$ as the ``folding'' rule  $\sketchR\Rightarrow\sketchL$ are sound in any LFOC. The universal ``unfolding'' rule for ``conditional existential quantification'' gives us a kind of ``\textbf{modus ponens}'' at hand \\
\((\typeX,
    \{(\typeX\arityof\expressiona,id_\typeX),
      (\typeX\arityof\condexistential{\expressiona}{\typemorphisma}{\typeY}{\expressionb},id_\typeX)
    \}
  )
  \stackrel{\contextmorphisma}{\Longrightarrow}
  (\typeY,
  \{(\typeY\arityof\expressionb,id_\typeY)\}).
\)

Other \textbf{universal sketch rules} depend only on the base category \katbase\ of a LFOC. As long as \katbase\ is a pre-sheaf topos, we do have, for example, universal rules at hand expressing reflexivity, symmetry and transitivity of ``equality'' (compare Example \ref{ex:universal-algebra:expressions}).
 
Besides ``universal rules'', we do have also \textbf{''(semantical) induced rules''}, i.e., sketch rules that are sound for all \metasig-structures we have chosen to be in $\katsemantics(\metasig)$. 

As a kind of \textbf{seventh parameter} of a LFOC, we can declare, the other way around,  a set of \textbf{''axiom rules''} where all axiom rules have to be sound for a \metasig-structure to be included \(\katsemantics(\metasig)\). In many cases  \(\katsemantics(\metasig)\) is defined exactly by all those \metasig-structures.
To declare a feature expression (''formula'') \(\typeX\arityof Exp\) as an \textbf{axiom}, in the traditional sense, we have to add a corresponding ``intro  rule'' \((\typeX,\emptyset)\Rightarrow(\typeX,\{(\typeX\arityof Exp,id_\typeX)\})\) to our axiom rules.
\end{remark}

\begin{example}[FOL: Sketch rules]
\label{ex:FOL:sketch-rules}
Horn clauses can be seen as axiom rules $\sketchL\Rightarrow\sketchR$ where the constraints in \sketchL\ and \sketchR\ use only atomic expressions.
A TBox in ALC is a finite set of ``terminological axioms'', i.e., of ``general concept inclusions'' \(\featuregen{C}\sqsubseteq\featuregen{D}\). The way, the semantics of ``general concept inclusions'' is defined in ALC, they correspond to ``axiom rules'' 
\hspace*{5em}\((\{p_1\},\{\constraint{\{p_1\}}{\featuregen{C}(p_1)}{id_{\{p_1\}}}\}) \Longrightarrow 
(\{p_1\},\{\constraint{\{p_1\}}{\featuregen{D}(p_1)}{id_{\{p_1\}}}\})	\).
\end{example}
\begin{example}[Category Theory: Sketch rules]
\label{ex:category-theory:sketch-rules}
There are, at least, two ways to axiomatize that all vertices do have an identity. We can reuse the expression \(\typeX\arityof\existsvia{}{\arity(\featureid)}{\featureid(pe)}\) from Example \ref{ex:category-theory:expressions}, where graph \typeX\ consists only of a vertex \(pv\), and  add the ``intro rule''
\((\typeX,\emptyset)
\Rightarrow
(\typeX,\{\constraint{\typeX}{\existsvia{}{\arity(\featureid)}{\featureid(pe)}}{id_{\typeX}}\})\) to our axiom rules. In this case, we need the universal ``modus ponens'' rule in Remark \ref{rem:universal-rules-axioms} to unfold the existence statement. 

Alternatively, we can declare existence of identities directly by adding an axiom rule \\
\hspace*{7em}
\((\typeX,\emptyset)\stackrel{\contextmorphisma}{\Longrightarrow}(\arity(\featureid),\{\constraint{\arity(\featureid)}{\featureid(pe)}{id_{\arity(\featureid)}}\})\), \\
with \contextmorphisma\ the inclusion of \typeX\ into \(\arity(\featureid)\). 

To require that identity morphisms are always unique, we declare the axiom rule \\
\hspace*{7em}\((\typeY,\{\constraint{\typeY}{\featureid(pe_1)\wedge\featureid(pe_2)}{id_\typeY}\})\stackrel{\contextmorphisma}{\Longrightarrow}(\arity(\featureid),\emptyset)\),\\
where \typeY\ is a graph with one node \(pv\) and two loops \(pe_1,pe_2\), and \(\substitutiona:\typeY\to\arity(\featureid)\) is given by the assignments \((pv\mapsto pv;pe_1,pe_2\mapsto pe)\).
\end{example}
\begin{example}[Universal Algebra: Sketch rules]
\label{ex:universal-algebra:sketch-rules}
Besides the universal rules expressing  reflexivity, symmetry and transitivity of ``term equality'', we do have available any rule induced by the fact, that our only chosen carrier is a finite product category, and by the way we defined the semantics of features in the only chosen structure. We can use rules stating, e.g., the existence of projections (variables considered as terms) or that ``term construction'' is total. We do have ``congruence rules'' for ``term equality'' and so on. 
\qed
\end{example}

Sketches represent, to a greater or lesser extent, properties of parts of our semantic structures and of the models, we are interested in. Sketch rules can be also used to appraise to what extend those properties are represented or have been made explicit in a sketch. 
\begin{definition}[Match]
\label{def:match}
A \textbf{match} of a \metasig-sketch  $\sketchG$ in a \metasig-sketch \sketchK\ is given by a context morphism \(\matcha:\contextG\to\contextK\) such that \(\functorCstr(\matcha)(\setofconstraintsG)\subseteq\setofconstraintsK\).
\end{definition}

\begin{definition}[Closedness]
\label{def:closedness}
A \metasig-sketch \sketchK\	is \textbf{closed under a \metasig-sketch rule} $\sketchL\stackrel{\contextmorphisma}{\Rightarrow}\sketchR$ \textbf{relative to a match} \(\matcha:\contextL\to\contextK\)  of the left-hand side \sketchL\  in \sketchK\ if, and only if, there exists a match \(\matchb:\contextR\to\contextK\) of the right-hand \sketchR\  of the rule in \sketchK\ such that \(\matcha=\contextmorphisma;\matchb\).
$$\xymatrix{
	\contextL \ar[rr]^\substitutiona\ar[dr]_(.3){\matcha\,:\,\functorCstr(\matcha)(\setofconstraintsL)\subseteq\setofconstraintsK}
	& {}\ar@{}[d]|(.4)=
	& \contextR  \ar[dl]^(.3){\exists\matchb\,:\,\functorCstr(\matchb)(\setofconstraintsR)\subseteq\setofconstraintsK}
	\\
	& \contextK
}$$
A \metasig-sketch \sketchK\	is \textbf{closed under a \metasig-sketch rule} $\sketchL\stackrel{\contextmorphisma}{\Rightarrow}\sketchR$ if, and only if, it is closed relative to each match \(\matcha:\contextL\to\contextK\)  of the left-hand side \sketchL\ in \sketchK.
\end{definition}
\begin{remark}[Sketch vs. structure]
\label{rem:sketch-vs-structure:equivalence}
Coming back to  our representation of ``elementary diagrams'' in  Remark \ref{rem:sketch-vs-structure}, we may say that the usefulness of ``elementary diagrams'' in model theory is based on the following \textbf{equivalence of sketches and structures}: 
A \metasig-structure $\structureU=(\carrierU,\setpred^\structureU)$ is conservative w.r.t.\ a \metasig-sketch rule  $\sketchL\stackrel{\contextmorphisma}{\Rightarrow}\sketchR$ if, and only if, the \metasig-sketch $\sketchS^\structureU= (\carrierU,C^\structureU)$  with 
\(C^\structureU=\{\constraint{\typeX}{\expression}{\interpretationa}\mid
\typeX\in\katvar,\,\interpretationa\in\semm{\expression}^\structureU_\typeX\}\) 
is closed under $\sketchL\stackrel{\contextmorphisma}{\Rightarrow}\sketchR$.
\\
\begin{minipage}[c]{.45\linewidth}
$$\xymatrix{
\contextL \ar[rr]^\substitutiona\ar[dr]_(.3){\contextinterpretation{\interpretationa}{\structureU} \models_\contextL\setofconstraintsL}
& {}\ar@{}[d]|(.4)=
& \contextR  \ar[dl]^(.3){\exists\contextinterpretation{\interpretationb}{\structureU} \models_\contextR \setofconstraintsR}
\\
& \carrierU
}$$
\end{minipage}
\begin{minipage}[c]{.55\linewidth}
$$\xymatrix{
\contextL \ar[rr]^\substitutiona\ar[dr]_(.3){\interpretationa\,:\,\functorCstr(\matcha)(\setofconstraintsL)\subseteq C^{\,\structureU}}
& {}\ar@{}[d]|(.4)=
& \contextR  \ar[dl]^(.3){\exists\interpretationb\,:\,\functorCstr(\matchb)(\setofconstraintsR)\subseteq C^{\,\structureU}}
\\
& \carrierU
}$$
\end{minipage}
\end{remark}

\begin{remark}[Deduction]
\label{rem:deduction}
In view of LFOC's, a bunch of deduction calculi can be described as procedures to construct sketches by means of rule applications, where the result of applying a sketch rule for a match of the left-hand side of the rule in a given sketch is constructed as a pushout in \(\katsketch(\metasig)\).
The application of a simple rule  $\sketchL\Rightarrow\sketchR$ will only add constraints to a given sketch, while the application of a rule $\sketchL\stackrel{\contextmorphisma}{\Rightarrow}\sketchR$, with \contextmorphisma\ non-isomorphic, will also extend and/or factorize the underlying context of a sketch. 

There are different parameters for such \textbf{''sketch based deduction calculi''} (and corresponding completeness results): (1) Kind of rules used as axiom rules. (2) Subset of the available universal rules used for deduction. (3) Kind of sketches serving as ``input'' for a deduction procedure. (4) Kind of sketches, we want to have as ``outputs''.


In the rest of the remark, we have to anticipate that our footprints declare also ``operations''.
We may characterize PROLOG by the following choices: Contexts are declarations of variables/individuals. (1) Horn clauses (2) None (?) (3)+(4) ``Facts''. If we consider Algebraic Specifications on the same abstraction level, as the FOL and Category Theory examples in this paper, we find a situation similar to PROLOG  (compare \cite{Rei87,Wol90,CGW95}):  Contexts are declarations of variables/generators. (1) Conditional equations (2) Reflexivity, symmetry, transitivity and congruence rules for ``term equality'' (3)+(4) Sets of equations.

A typical problem, as in Algebraic Specifications, for example, is to deduce for a given set of axiom rules all semantical induced rules, where \(\katsemantics(\metasig)\) is given by all structures such that all axiom rules are sound for them. In Algebraic Specifications, we can solve this problem by means of a sketch based deduction calculus that derives sets of equations from given sets of equations \cite{Rei87,Wol90}. At the moment, it is open for us to what extent this kind of ``deduction theorem'' can be generalized to LFOC's. For us it is very hard (and in some periods even impossible) to approach the peculiar notation as well as the high level of abstraction in \cite{Mak97}, but there should be some hints in this direction. 

Resolution in PROLOG can be seen as a procedure that derives new Horn clauses from given ones. An analogous procedure, called ``parallel resolution'', that allows to deduce conditional equations directly from given ones, is presented in \cite{Wol90}. This procedure is  sound and complete (compare Theorem 5.2.4 in \cite{Wol90}). We hope, that we can generalize those kinds of ``resolution procedures'', at least to a certain extent, to sketch rules in LFOC's.

A \textbf{completeness proof} for a ``sketched based deduction calculus'' may be done in the way, we have essentially done it in \cite{Wol90}: We show that the calculus allows to construct ``freely generated sketches'' that are closed under the relevant rules. Then we show, that those closed sketches can be transformed into structures that are equivalent to the generated sketches, in the spirit of Remark \ref{rem:sketch-vs-structure:equivalence}, and inherit a kind of ``freely generated'' property from the ``freely generated sketches''. 
\qed
\end{remark}

\section{Conclusions and Future Research}
Summing up many experiences, insights, results and ideas from our research in the wider area of formal specifications, we developed in this paper the first basic building block of a framework of Logics of First-Order Constraints (LFOC's). As a sanity check for this new approach to logic, we presented an abstract uniform scheme how to define a certain LFOC in such a way that we get an institution in the sense of \cite{GB92,Dia08}. To a certain extend the paper can be seen as a proposal for a bigger and broader research project aiming at to develop a fully shaped general framework of Logics of First-Order Constraints.
 
There are many open ends, problems, questions and ideas we would like to address within such a research project in the future. We list only few of them:
\begin{itemize}
	\item Substitution framework in LFOC's where \katbase\ has pushouts, especially, for the base categories \katgraph\ and \katGRAPH.
	\item Full account of universal ``folding'' and ``unfolding'' rules.
	\item Deduction calculi (see Remark \ref{rem:deduction}).
	\item Footprints with operations thus substitutions become Kleisli-morphisms \cite{WolterDK18}. 
	\item Restrictions of LFOC's in the spirit of Description Logics.
	\item Semantic-as-instance approach.
	\item Typed contexts and language extensions
	\item Dependencies between features in generalizing the approach and results in \cite{DW08}.
	\item Precise connection between different abstraction levels. Here the semantic-as-instance approach may be most appropriate?
	\item Different relevant restrictions for the construction of feature expressions. See, for example the Open Issue \ref{open:constructive-feature-expressions} (Constructive Feature expressions). 
	\item Relation to nested graph conditions \cite{Rensink04}.
	\item Working out a bunch of examples. 
	\item $\bullet\quad\bullet\quad\bullet$
\end{itemize}



\end{document}